%% file: main.tex
\documentclass{lmcs} %%% last changed 2014-08-20
\pdfoutput=1

\usepackage[utf8]{inputenc}

\usepackage{lastpage}
\lmcsdoi{20}{2}{1}
\lmcsheading{}{\pageref{LastPage}}{}{}%
{Mar.~07,~2023}{Apr.~04,~2024}{}


\usepackage{cat-doc} 
\usepackage{bm}

\usepackage[hidelinks]{hyperref} 
  \usepackage[capitalize]{cleveref} 

  \crefname{figure}{Figure}{Figures}
  \crefname{thm}{Theorem}{Theorems}
  \crefname{lem}{Lemma}{Lemmas}
  \crefname{defi}{Definition}{Definitions}
  \crefname{prop}{Proposition}{Propositions}
  \crefname{cor}{Corollary}{Corollaries}
  \crefname{exa}{Example}{Examples}
  \crefname{rem}{Remark}{Remarks}
  \crefname{propC}{Proposition}{Propositions}

\begin{document}

\title[A Fibrational Tale of Operational Logical Relations]{A Fibrational Tale of Operational Logical\texorpdfstring{\\}{} Relations: Pure, Effectful and Differential} %Please add

%\titlecomment{{\lsuper*}OPTIONAL comment concerning the title, \eg, 
%  if a variant or an extended abstract of the paper has appeared elsewhere.}

\author[F.~Dagnino]{Francesco Dagnino\lmcsorcid{0000-0003-3599-3535}}[a]	%required
\author[F.~Gavazzo]{Francesco Gavazzo\lmcsorcid{0000-0002-2159-0615}}[b]	%required

\address{DIBRIS, Università di Genova, Italy}	%required
\email{francesco.dagnino@dibris.unige.it}  %optional
%\thanks{thanks 1, optional.}	%optional

\address{Università di Padova, Italy}	%required
\email{francesco.gavazzo@unipd.it}  %optional
%\thanks{thanks 1, optional.}	%optional

%% etc.

%% required for running head on odd and even pages, use suitable
%% abbreviations in case of long titles and many authors:

%%%%%%%%%%%%%%%%%%%%%%%%%%%%%%%%%%%%%%%%%%%%%%%%%%%%%%%%%%%%%%%%%%%%%%%%%%%

%% the abstract has to PRECEDE the command \maketitle:
%% be sure not to issue the \maketitle command twice!

\begin{abstract}
Logical relations built on top of an operational semantics are one of the most successful proof methods in programming language semantics. 
In recent years, more and more expressive notions of operationally-based logical relations have been designed and applied to specific families of languages. 
However, a unifying abstract framework for operationally-based logical relations is still missing. 
We show how fibrations can provide a uniform treatment of operational logical relations, using as reference example 
a $\lambda$-calculus with generic effects endowed with a novel, abstract operational semantics defined on a large class of categories.
Moreover, this abstract perspective allows us to give a solid mathematical ground also to differential logical relations --- a recently 
introduced notion of higher-order distance between programs --- both pure and effectful,
bringing them back to a common picture with traditional ones. 
\end{abstract}

\maketitle

\input{introduction}

\input{operational}

\input{prelim}

\input{dif-rel}

\input{conclu}

% \section*{Acknowledgment}
%
%
%% in general the use of bibtex is encouraged

\bibliographystyle{alphaurl}
\bibliography{biblio}

% \appendix 

\end{document}

%% file: introduction.tex
\section{Introduction}
\label{sect:introduction} 

Logical relations \cite{Reynolds/Logical-relations/1983} are one of the most successful proof techniques in logic and programming language semantics. 
Introduced in proof theory \cite{DBLP:journals/jsyml/Tait67,girard1989proofs} in their unary form, logical relations quickly became a main tool 
in programming language semantics. In fact, starting with the seminal work by Reynolds \cite{Reynolds/Logical-relations/1983},
Plotkin \cite{Plotkin-Lambda-definability-logical-relations}, and Statman \cite{DBLP:journals/iandc/Statman85},
logical relations have been extensively used to study both the denotational
and operational behaviour of programs.\footnote{See the classic textbooks by Mitchell \cite{mitchell}, Pierce \cite{pierce/types/and/programming/languages}, and Harper \cite{harper} (and references 
therein) for an introduction to both denotationally- and operationally-based logical relations.}

Logical relations (and predicates) mostly come in two flavours, depending on whether they are defined relying on 
the \emph{operational} or \emph{denotational} semantics of a language. We refer to logical relations of the first kind as 
\emph{operational logical relations} and to logical relations of the second kind as 
\emph{denotational logical relations}. 
Due to their link with denotational semantics, denotational logical relations have 
been extensively studied in the last decades, both for specific programming languages and in the abstract, 
this way leading to beautiful general theories of 
(logical) predicates and relations on program denotations. 
In particular, starting with the work by Reynolds and Ma \cite{DBLP:conf/mfps/MaR91}, Mitchell and Scedrov \cite{scedrov}, 
and Hermida \cite{Hermida93}, researchers have started to investigate notions of (logical) predicates and relations 
in a general categorical setting, this way giving rise to an abstract understanding of relations over (the denotational 
semantics of) programs centered around the notions of \emph{fibrations} 
\cite{Hermida93,johann1, johann2,DBLP:journals/lmcs/KatsumataSU18,DBLP:journals/iandc/Katsumata13,DBLP:conf/popl/Katsumata14,DBLP:conf/csl/Katsumata05}
\emph{reflexive graphs} \cite{OHearnT95,RobinsonR94,reddy,HermidaRR14}, and 
\emph{factorisation systems} \cite{DBLP:journals/entcs/KammarM18,DBLP:journals/entcs/HughesJ02,GoubaultLasotaNowak/MSCS/2008}. 
The byproduct of all of that is a general, 
highly modular
theory of denotational logical relations that has been successfully applied to a large array of language 
features, ranging from parametricity and polymorphism \cite{johann1, johann2, reddy} to computational effects 
\cite{DBLP:journals/lmcs/KatsumataSU18,DBLP:journals/iandc/Katsumata13,DBLP:conf/popl/Katsumata14,DBLP:conf/csl/Katsumata05, 10.1007/978-3-319-71237-6_17,DBLP:journals/entcs/KammarM18,GoubaultLasotaNowak/MSCS/2008}.

On the operational side, researchers have focused more on the development and applications of 
expressive notions of logical relations for specific (families of) languages, rather than on their underlying theory. In fact, 
operational logical relations being based on operational semantics, they can be easily defined for languages  
for which finding the right denotational model is difficult, this way making operational logical relations 
a handy and lightweight technique, especially when compared to its denotational counterpart. As a paradigmatic example, consider 
the case of stochastic $\lambda$-calculi and their operational techniques 
\cite{DBLP:conf/icfp/BorgstromLGS16,dal-lago/gavazzo-mfps-2019,DBLP:journals/pacmpl/WandCGC18} which can be (easily) defined 
relying on the category of measurable spaces and measurable functions, but whose denotational semantics have
required the introduction of highly nontrivial mathematical structures \cite{StatonYWHK16,DBLP:conf/esop/Staton17,DBLP:journals/pacmpl/VakarKS19,EhrhardPT18}, 
since the category of measurable spaces (and measurable functions) is not closed \cite{aumann}. 

The wide applicability of operational logical relations, however, has also prevented the latter 
to organise as a uniform \emph{corpus} of techniques with a common underlying theory. Operational logical relations 
result in a mosaic of powerful techniques applied to a variety of languages --- 
including higher-order, functional, imperative, and concurrent languages \cite{DBLP:journals/corr/abs-1103-0510, DBLP:conf/esop/Ahmed06, DBLP:conf/popl/TuronTABD13,AppelMcAllester/TOPLAS/2001,DBLP:journals/entcs/CraryH07,DBLP:journals/jfp/DreyerNB12};
both pure and (co)effectful \cite{JohannSimpsonVoigtlander/LICS/2010,DalLagoG22,DBLP:journals/pacmpl/AbelB20,DBLP:journals/iandc/BirkedalJST16,DBLP:conf/birthday/Hofmann15,DBLP:conf/popl/Benton0N14,DBLP:journals/pacmpl/BiernackiPPS18,BizjakBirkedal/FOSSACS/2015} --- 
whose relationship, however, is unclear.
This may not be a prime concern of the semanticist interested in somehow conventional program 
behaviours, such as termination, extensional equivalence, cost, etc.\footnote{Although the absence of a general 
operational foundation makes logical relations heavily syntax-dependent, this way preventing 
the development of a modular and compositional theory of operational logical relations.} 
However, it becomes so when one faces logical relations for program behaviours that are not yet fully understood, 
such as those akin to metric~\cite{Pierce/DistanceMakesTypesGrowStronger/2010,DBLP:conf/lics/Pistone21} and differential properties of programs
\cite{DBLP:conf/icalp/LagoGY19,dlrICTCS,dlrTCS,DalLagoG22}.
In those cases, in fact, logical relations lose their relational core and are modelled through
more articulated structures that may go significantly far from classic relations and generalisations thereof. 

A paradigmatic case is the one of the so-called \emph{differential logical relations}~\cite{DBLP:conf/icalp/LagoGY19,dlrICTCS,dlrTCS,DalLagoG22}, 
a recently introduced foundational framework for differential (higher-order) program semantics. Differential logical relations 
relate pairs of programs with so-called \emph{higher-order distances}, namely notions of distances that can interact with the related programs 
in a parameter-passing fashion. For instance, differential logical relations measure the distance between functional programs $P$, $Q$ not as a number, as 
in ordinary metric reasoning, but as a (possibly higher-order) \emph{function} $\delta$; and the `logical' behaviour 
of $\delta$ prescribes that on any pair of inputs $V, W$ with 
(higher-order) distance $\varepsilon$, the distance between $P V$ and $Q W$ is given by $\delta(V,\varepsilon)$. 
That is, the distance $\delta$ is itself a function that, given two inputs and a distance between them, 
produce an input-dependent distance between outputs. 

Differential logical relations indeed share some commonalities with traditional logical relations, but it is not mathematically clear 
what is the exact relationship between those two worlds: for instance, are differential logical relations indeed logical relations? (And, if so, 
in which sense?) Is the so-called fundamental lemma of differential logical relations a true instance of the classic fundamental lemma of logical relations? And, 
more generally, how much of the theory of ordinary logical relations can be exported to the realm of differential logical relations, and vice versa? 
All those questions become even harder to answer if one considers \emph{effectful} extensions of differential logical
relations~\cite{DalLagoG22}. Such extensions, in fact, mix traditional monadic reasoning \emph{\`a la} Moggi~\cite{Moggi/LICS/89} 
together with the novel notion of a \emph{differential extension}.  The mathematical status of such a notion is currently unknown, in the sense 
that it is not known whether a differential extension could be seen as a relational 
extension~\cite{Barr/LMM/1970,Kurz/Tutorial-relation-lifting/2016,Thijs/PhDThesis/1996,Carboni/Bicategories-and-spans/1984} 
generalised to a differential setting, nor whether the `canonical' differential extension used to concretely define 
effectful logical relations is indeed canonical; not to mention the status of the fundamental lemma of 
effectful differential logical relations.

Answering those questions is not only a crucial step towards the development of a solid theory of differential semantics throughout 
(effectful) differential logical relations, but it is also key to improving the applicability of the latter. In fact, when reasoning about differential 
behaviours, such as input-dependent cost analysis or approximate correctness, it appears useful to mix differential and non-differential 
reasoning~\cite{DalLagoG22}, so that it is desirable to view non-differential reasoning as a collapsed form of differential one. 
Concretely, that means recovering traditional logical relations as differential ones with rather trivial (higher-order) distances.
%in such a way that compositional properties of differential reasoning imply similar properties in the non-differential setting 
%(e.g. the fundamental lemma of differential logical relations should imply the traditional fundamental lemma). 
However, the current theory of differential logical relations does not allow one to do that, unless one shows 
from scratch how to encode a \emph{specific} family of logical relations
as differential logical relations. 

% problem as long as one 
% This situation creates a peculiar scenario where, on the one hand, 
% the effectiveness of operational logical relations has been proved by their many applications but, on the other hand, 
% a foundational understanding of operational logical relations is still missing, the main consequence of that being the 
% lack of modularity in their development. 
% All of that becomes even worst if one also takes into account more recent forms of logical relations, such as metric 
% \cite{Pierce/DistanceMakesTypesGrowStronger/2010,DBLP:conf/lics/Pistone21} and 
% differential \cite{dlrICALP,dlrICTCS,dlrTCS,DalLagoG22} ones, that go beyond traditional relational reasoning.

In this paper, we find a remedy for all of that and show that (effectful) differential logical relations are 
a form of operational logical relations in a very precise mathematical sense. We do more, actually. 
In fact, to prove the aforementioned result, we first need a general definition of what an operational logical relation is,
something that, as already remarked, seems unavailable in the literature.
We fix that by taking a \emph{fibrational foundation} of operational logical relations, 
hence showing that much in the same way denotational logical relations can be uniformly understood in terms 
of fibrations, it is possible to give a uniform account of \emph{operational logical relations} relying on the 
language of \emph{fibrations}. 
Armed with such a fibrational understanding, we then show how differential logical 
relations --- both \emph{pure} and \emph{effectful} ---  can be seen as a suitable instance of our fibrational framework. 
Remarkably, the level of abstraction offered by fibrations allows us to come up with a general fundamental 
lemma of logical relations that covers both traditional and differential instances of the homonymous result. 

\subsection{Contribution}
The main contribution of this paper is twofold, covering both the general theory of 
operational logical relations and the specific theory of pure and effectful differential 
logical relations. 

\subsubsection{Operational Logical Relations, Fibrationally}
Our first contribution is the development of a general, abstract notion of an operational logical relation 
in terms of fibrations for a $\lambda$-calculus with generic effects. 
Fibrations are a mainstream formalism for general, categorical notions of predicates/relations. 
More precisely, a fibration is a suitable functor from a category of (abstract) predicates --- the domain of the fibration ---
to a category of arguments --- usually called the base category. 
In denotational logical relations, predicates usually apply to program denotations rather than to programs themselves, 
with the main consequence that the base category is usually required to be cartesian closed. 
In this paper, we follow a different path and work with base categories describing the operational (and interactive) 
behaviour of programs, rather than their denotations. 
To do so, we introduce the novel notion of an \emph{operational structure}, 
the latter being a cartesian category with arrows describing (monadic) evaluation semantics \cite{PlotkinPower/FOSSACS/01,DalLagoGavazzoLevy/LICS/2017} and satisfying suitable coherence conditions encoding 
the base dynamics of program evaluation. This way, we give not only an abstract account of traditional set-based evaluation 
semantics, but also of evaluation semantics going beyond the category of sets and functions, the prime example 
being the evaluation semantics of stochastic $\lambda$-calculi, which is defined as a stochastic kernel \cite{dal-lago/gavazzo-mfps-2019,DBLP:journals/pacmpl/WandCGC18}.

On top of our abstract operational semantics, we give a general notion of an operational logical relation in terms of (logical) 
fibrations and prove a general result (which we call the \emph{fundamental lemma} of logical relations, following the standard 
nomenclature of concrete, operational logical relations) stating that programs behave as arrows in the domain of the fibration. 
Remarkably, our general fundamental lemma subsumes several concrete instances of the fundamental lemma of logical relations 
appearing in the literature. Additionally, the operational nature of our framework immediately 
results in a wide applicability of our results, especially if compared with fibrational accounts of denotational 
logical relations. In particular, since our logical relations build upon operational structures, 
they can be instantiated to non-cartesian-closed categories, this way reflecting at a general level the wider applicability of 
operational techniques with respect to denotational ones. As a prime example of that, 
we obtain operational logical relations (and their fundamental lemma) for stochastic $\lambda$-calculi for free, 
something that is simply not achievable denotationally, due to the failure of cartesian closedness of the category 
of measurable spaces.

\subsubsection{Fibrational Differential Reasoning}
Our second, main contribution is to show that our framework goes beyond traditional relational reasoning, 
as it gives a novel mathematical account of 
the recently introduced differential logical relations \cite{DBLP:conf/icalp/LagoGY19,dlrICTCS,dlrTCS,DalLagoG22} (DLRs,  for short), 
both pure and effectful. DLRs are a new form of logical 
relations introduced to define higher-order distances between programs, such distances being abstract notions 
reflecting the interactive complexity of the programs compared.  
DLRs have been studied \emph{operationally} and on specific calculi only, oftentimes introducing 
new notions --- such as the one of a \emph{differential extension} of a monad \cite{DalLagoG22} --- whose mathematical status is still 
not well understood. The main consequence of that is that
a general, structural account of pure and effectful DLRs is still missing. 
In this paper, we show how DLRs are a specific instance of our abstract operational logical relations 
and how the fundamental lemma of DLRs is an instance of our general fundamental lemma. 
We do so by introducing the novel construction of a \emph{fibration of differential relations} and showing how the latter precisely captures 
the essence of DLRs, bringing them back to a common framework with traditional logical relations. 
Additionally, we show how
our fibrational account sheds new light on the mathematical status of effectful DLRs. In particular, 
we show that differential 
extensions of monads are precisely liftings of monads 
to the fibration of differential relations, and that the so-called coupling-based 
differential extension \cite{DalLagoG22} --- whose canonicity has been left as an open problem ---
is an instance of a general monadic lifting to the fibration of differential relations: remarkably, such a lifting
is the extension of the well-known Barr lifting \cite{Barr/LMM/1970} to a differential setting.

\subsubsection{Source of the material}
This is an extended and revised version of a paper that appears in the proceedings of FSCD 2022 \cite{DagninoG22}. 
The present version provides more background on Differential Logical Relations to better frame the work in the literature, 
includes all proofs of our results and reports some new examples.

\subsection{Related Work}
Starting with the seminal work by Hermida \cite{Hermida93}, fibrations have been used 
to give categorical notions of (logical) predicates and relations, and to model denotational logical relations \cite{johann1, johann2,DBLP:journals/lmcs/KatsumataSU18,DBLP:journals/iandc/Katsumata13,DBLP:conf/popl/Katsumata14,DBLP:conf/csl/Katsumata05}, as they provide a formal way to relate 
the denotational semantics of a programming language and a logic for reasoning about it. 
Other categorical approaches to denotational logical relations have been given in terms of 
reflexive graphs \cite{OHearnT95,RobinsonR94,reddy,HermidaRR14} and 
factorisation systems \cite{DBLP:journals/entcs/KammarM18,DBLP:journals/entcs/HughesJ02,GoubaultLasotaNowak/MSCS/2008}. 

On the operational side, fibrations have been used to give abstract accounts to induction and coinduction \cite{DBLP:journals/iandc/HermidaJ98}, 
both in the setting of initial algebra-final coalgebra semantics \cite{johann3,johann4,johann5} and in the setting of up-to techniques \cite{DBLP:conf/csl/BonchiPPR14,DBLP:conf/concur/Bonchi0P18}. To the best of the authors' knowledge, however, none of these approaches has been 
applied to operational reasoning for higher-order programming languages. 
Concerning the latter, general operational accounts of logical relations both for effectful \cite{JohannSimpsonVoigtlander/LICS/2010} and 
combined effectful and coeffectful languages \cite{DBLP:journals/pacmpl/AbelB20,DBLP:journals/pacmpl/LagoG22a} have been given 
in terms of relational reasoning. These approaches, however, are tailored to specific operational semantics and notions of relations, 
and thus they cannot be considered truly general. 
Finally, DLRs have been studied mostly operationally \cite{DBLP:conf/icalp/LagoGY19,dlrICTCS,dlrTCS,DalLagoG22}, 
although some general \emph{denotational} accounts of DLRs have been proposed \cite{DBLP:conf/icalp/LagoGY19,DBLP:conf/lics/Pistone21}. 
Even if not dealing with operational aspects of DLRs, the latter proposals can cope with \emph{pure} DLRs only, and are  
too restrictive to incorporate computational effects.

%% file: operational.tex
\section{The Anatomy of an Operational Semantics}
\label{sect:operational} 

To define a general notion of an operational logical relation, we first need to 
define a general notion of an operational semantics. This is precisely the purpose of 
this section.
In particular, we introduce the notion of an \emph{operational structure} on a category 
with finite products endowed with a strong monad 
as an axiomatisation of a general evaluation semantics. 
Operational structures prescribe the existence of basic interaction arrows (the latter describing basic 
program interactions as given by the usual reduction rules) 
and define program execution as a Kleisli arrow (this way giving monadic evaluation) 
satisfying suitable coherence 
laws reflecting evaluation dynamics. 
Remarkably, operational structures turn out to be more liberal --- hence widely applicable --- 
than categories used in denotational semantics 
(the latter being required to be cartesian closed). 

% we show how to use fibrations to give a uniform account 
% to \emph{operationally-based} logical relations, both for pure and effectful 
% languages. Contrary to denotationally-based logical relations, we use the base category 
% to specify how an external observer interacts with programs, such interaction been 
% genuinely operational. For instance, instead of interpreting programs as denotational 
% objects (which require case categories to be cartesian-closed), we treat them as 
% syntactic objects and simply requires \emph{operational} semantics to be an arrow 
% in the base category, this way reflecting the way an external observer interacts 
% with a program, i.e. evaluating it. Much in the same way, we will require the base category 
% to come with maps performing, e.g., value-passing and tuple decomposition, the latter being 
% nothing but the ways $\lambda$-abstractions and pairs are tested by external observers. 
% Requiring base categories to come with structure to describe the interactive behaviour of 
% programs, rather then their denotations, considerably improves the applicability of our framework, 
% that can be instantiated to categories that do not naturally support denotational reasoning 
% (not being cartesian closed), but are nonetheless rich enough to define 
% operational techniques. As an example of that, we will show how to define logical relations 
% for a stochastic $\lambda$-calculus with continuous probability distributions in the category 
% of measurable spaces. 

\subsection{A Calculus with Generic Effects} 

Our target calculus is a simply-typed fine-grained call-by-value \cite{Levy/InfComp/2003}
$\lambda$-calculus, denoted by $\lambdacalc$,
enriched with generic effects \cite{PlotkinPower/FOSSACS/01,Plotkin/algebraic-operations-and-generic-effects/2003}. 
The syntax of $\lambdacalc$ is given in Figure~\ref{figure:typing-lambda}: to explain that, we need to 
recall some preliminary notions. 

We assume familiarity with monads and 
algebraic effects~\cite{Moggi/LICS/89,PlotkinPower/FOSSACS/01}. 
In particular, we use the notation $\mnd = (\mndfun, \mndun, \mndmul, \str)$ for strong monads on 
cartesian categories, 
recalling that the latter notion is equivalent to the one of a strong Kleisli triple
$\mnd = (\mndfun, \mndun, \bindsymbol)$, where 
for an arrow $f: X \times Y \to \mndfun Z$ in a cartesian category $\CC$, we denote by 
$\bindsymbol f: X \times \mndfun Y \to \mndfun Z$ the strong Kleisli extension of $f$. 
We tacitly use strong monads  
and strong Kleisli triples interchangeably, and
write $\kleisli{f}: \mndfun X \to \mndfun Y$ for the Kleisli 
extension of $f: X \to \mndfun Y$.\footnote{Formally, $\kleisli{f}$ is defined as $\bindsymbol f'$, where 
$f': 1 \times X \to \mndfun Y$ is obtained from $f$ via the isomorphism 
$1 \times X \cong X$.}
Given a strong monad
$\mnd = (\mndfun, \mndun, \bindsymbol)$ on a cartesian category $\CC$, a generic effect \cite{Plotkin/algebraic-operations-and-generic-effects/2003,PlotkinPower/FOSSACS/01} of arity $A$, with $A$ an object of $\CC$, 
is an arrow $\gopsem: 1 \to \mndfun A$. Standard examples of generic effects are obtained by taking 
$\CC = \Set$ and $A$ equal to a finite set giving 
the arity of the effect. For instance, one models nondeterministic (resp. fair) coins as elements of $\mathcal{P}(2)$ 
(resp. $\mathcal{D}(2)$), 
where $\mathcal{P}$ (resp. $\mathcal{D}$) is the powerset (resp. distribution) monad. 
Other examples of generic effects include primitives for input-output, memory updates, exceptions, etc.
Here, we assume we have a collection of generic effect symbols $\gop$, each with an associated type $\typeone_{\gop}$. 
We leave the interpretation of $\gop$ as an actual generic effect $\gopsem$ to the operational semantics.

Armed with these notions, we give the syntax and static semantics of $\lambdacalc$ 
in Figure~\ref{figure:typing-lambda}, where $\basetype$ ranges over base types 
and $\cbase$ over constants of type $\basetype$ (for ease of exposition, we do not include operations on base types, although  
those can be easily added). 

\begin{figure*}[htbp]
%	\footnotesize
	\hrule
	\begin{align*}
    \typeone, \typetwo &::= \zeta \mid \sigma_\gamma \mid \unittype \mid \sigma \times \tau \mid \sigma \to \tau
    \\
    \valone, \valtwo &::= \varone \mid \cbase \mid \unitval \mid \abs{\varone}{\termone} \mid \pair{\valone}{\valtwo}
    \\
    \termone, \termtwo &::= \return{\valone} \mid \valone \valtwo \mid \fst{\valone} \mid 
    \snd{\valone} \mid \seq{\termone}{\termtwo} \mid \gop
\end{align*}

%\hrule
	\begin{center}
	\[
	\infer{\envone \valimp \varone: 
	\typeone}{\varone: \typeone \in \Gamma}
	\quad 
	\infer{\envone \valimp \cbase: \basetype}{}
	\quad 
	\infer{\envone \compimp \return{\valone}: \typeone}{\envone \valimp \valone: 
	\typeone}
	\quad 
% 	\infer{\envone \compimp \opsymbol(\termone_1, \hh, \termone_n): \typeone}
% 	{\envone \compimp \termone_1: \typeone & \hh &  \envone \compimp \termone_n: \typeone}
    \infer{\envone \compimp \gop: \typeone_{\gop}}
	{\phantom{\varone}}
	\quad	\hspace{-0.1cm}
	\infer{\envone \valimp \abs{\varone}{\termone}: \typeone\to \typetwo}
	{\envone, \varone:\typeone \compimp \termone: \typetwo}
	%%%%%%%%%%%%%%%
	\quad \hspace{-0.1cm}
	%%%%%%%%%%%%%%% 
	\infer{\envone \compimp \valone\valtwo: \typetwo}
	{\envone \valimp \valone : \typeone \to \typetwo
		&
		\envone \valimp \valtwo: \typeone}
	\]
	$\vspace{-0.2cm}$
	\[
	\infer{\envone \compimp \seq{\termone}{\termtwo}: \typetwo}
		{\envone \compimp \termone : \typeone 
			&
			\envone, \varone: \typeone \compimp \termtwo: \typetwo}
	\quad
	%%%%%%%%%%%%%%%%%%
	\infer{\envone \compimp \fst{\valone}: \typeone}
	{\envone \valimp \valone: \typeone \times \typetwo}
	%%%%%%%%%%%%
	\quad
	%%%%%%%%%%%%%%
		\infer{\envone \compimp \snd{\valone}: \typetwo}
	{\envone \valimp \valone: \typeone \times \typetwo}
	\quad 
	\infer{\envone \valimp \pair{\valone}{\valtwo}: \typeone \times \typetwo}
	{\envone \valimp \valone: \typeone
	&
	\envone \valimp \valtwo: \typetwo}
	\quad 
	\infer{\envone \valimp \unitval: \unittype}{}
	\]
	$\vspace{-0.2cm}$
	\end{center}
	\hrule
	\caption{Syntax and static semantics of $\lambdacalc$}
	\label{figure:typing-lambda}
\end{figure*}
%	\infer{\envone \valimp \makereal{r}: \typereal}{r \in \mathbb{R}}
%	\qquad 
%	\infer{\envone \compimp \sample: \typereal}{\phantom{\varone} }
%	\qquad 
%	\infer
%	{\envone \compimp \makereal{\op}(\valone_1, \hh, \valone_n): \typereal}
%	{\envone \valimp \valone_1 : \typereal, \hh,
%	\envone \valimp \valone_n: \typereal
%	& \op \in \mathcal{C}_n}

The notation used for the static semantics is mostly standard: a typing judgment
is an expression of the form $\Gamma \imp t: \sigma$ or $\Gamma \imp v: \sigma$, where 
$\envone$ is an environment, i.e. 
a set of pairs $x: \sigma$ of variables and types such that if both $x:\sigma$ and $x:\tau$ belongs to $\Gamma$, 
the $\sigma = \tau$, and $v$ and $t$ are a value and a computation, respectively.
Indeed, expressions of $\lambdacalc$ are divided into two (disjoint) classes: \emph{values} 
(notation $\valone, \valtwo, \hdots$) and \emph{computations} (notation $\termone, \termtwo, \hdots$), 
the former being the result of a
computation, and the latter being an expression
that once evaluated may produce a value (the evaluation process might
not terminate) as well as side effects. 
When the distinction between values and computations is 
not relevant, we generically refer to \emph{terms} 
(and still denote them as $\termone, \termtwo, \hdots$).
% We use judgements $\envone \valimp \valone: \typeone$ 
% and $\envone \compimp \termone: \typeone$ to indicate that $\valone$ and $\termone$ 
% are values and computation of type $\typeone$ in environment $\envone$, respectively.
We adopt standard syntactic conventions \cite{Barendregt/Book/1984} and 
identify
terms up to renaming of bound variables: we say that a term is closed
if it has no free variables and write $\values_{\typeone}$, $\terms_{\typeone}$ 
for the sets of closed values and computations of type $\typeone$, respectively.
We write $\subst{\termone}{\valone_1, \hdots, \valone_n}{\varone_1, \hdots, \varone_n}$ 
(and similarly for values)
for the capture-avoiding (simultaneous) substitution of the values $\valone_1, \hdots, \valone_n$
for all free occurrences of $\varone_1, \hdots, \varone_n$ in $\termone$. Oftentimes, we will use the 
notation $\vect{\phi}$ for a sequence $\phi_1, \hdots, \phi_n$ of symbols $\phi_i$, for 
$i \in \{1, \hdots, n\}$.

When dealing with denotational logical relations, one often organises $\lambdacalc$ 
as a syntactic category having types as objects and (open) terms modulo the usual $\beta\eta$-equations as arrows. 
Operationally, however, terms are purely syntactic objects and cannot be taken modulo $\beta\eta$-equality. 
For that reason, we consider a \emph{syntactic graph} rather than 
a syntactic category.\footnote{Recall that a graph is defined by removing from the definition of a category the axioms prescribing 
the existence of identity and composition. A diagram is a map from graphs to graphs (a diagram being defined by removing 
from the definition of a functor the clauses on identity and composition). Since any category is a graph, 
we use the word \emph{diagram} also to denote maps from graphs to categories.}

\begin{defi}
The objects of the syntactic graph $\syn$ are environments $\envone$, types $\typeone$, and expressions 
$\comptype{\typeone}$, for  
$\typeone$ a type. Arrows are defined thus: $\hom(\envone, \typeone)$ 
consists of values $\envone \valimp \valone: \typeone$, whereas
$\hom(\envone, \comptype{\typeone})$ consists of terms 
$\envone \compimp \termone: \typeone$; otherwise, the hom-set is empty.
\end{defi}

The definition of $\syn$ reflects the call-by-value nature of $\lambdacalc$: to each type $\typeone$ 
we associate two objects, representing the type $\typeone$ on values and on computations. Moreover, 
there is no arrow having environments as codomains nor having objects $\comptype{\typeone}$
as domain: this reflects 
that in call-by-value calculi variables are placeholders for values, not for computations. 

% \begin{rem}
% $\syn$ consider arrows from single types to single types only, so that arrows are given by expression 
% with a single free variable. This choice simplifies the exposition of our techniques and none of our results 
% depend on that: everything we prove in this section extend \emph{mutatis mutandins} to the extension of 
% $\syn$ where one considers tuples of types as objects and open terms and values 
% $\varone_1: \typeone_1, \hh, \varone_n: \typeone_n \valimp \valone: \typeone$, 
% $\varone_1: \typeone_1, \hh, \varone_n: \typeone_n \compimp \termone: \typeone$ as arrows 
% from $(\typeone_{1}, \hh, \typeone_{n}) \to (\typeone)$ and 
% $(\typeone_{1}, \hh, \typeone_{n}) \to (\comptype{}{\typeone})$, respectively. 
% \end{rem}

\subsection{Operational Semantics: The Theoretical Minimum}
Having defined the syntax of $\lambdacalc$, we move to its operational semantics. 
Among the many styles of operational semantics (small-step, big-step, etc.), evaluation semantics 
turns out to be a convenient choice for our goals. 
Evaluation semantics are usually defined as monadic functions $\ev: \terms_{\typeone} \to \mndfun(\values_\typeone)$, 
with $\mndfun$ a monad encoding the possible effects produced during program evaluation 
(e.g., divergence or nondeterminism) \cite{DalLagoGavazzoLevy/LICS/2017,PlotkinPower/FOSSACS/01,JohannSimpsonVoigtlander/LICS/2010}. 
To clarify the concept, let us consider an example from \cite{DalLagoGavazzoLevy/LICS/2017}.

\begin{exa}
\label{ex:algebraic-effects}
Let $\mnd = (T, \eta, \mu)$ be a monad on \Set with a generic effects $\gopsem \in \mndfun(\values_{\typeone_{\gop}})$ for 
any effect symbol $\gop$ in $\lambdacalc$. The (monadic) evaluation (family of) map(s) 
$\eval{-}: \terms_{\typeone} \to \mndfun(\values_{\typeone})$ is defined as follows (notice that $\lambdacalc$ being simply-typed, $\sem{-}$ is well-defined\footnote{
It is worth remarking that termination is not an issue for 
our general notion of an evaluation semantics (Definition~\ref{def:evaluation-structure} below). 
In fact, the results presented in this paper simply require having an evaluation map 
satisfying suitable coherence conditions. Consequently, we could rephrase 
this example as to ignore termination by requiring the monad to be enriched in a $\omega$-complete 
partial order \cite{AbramskyJung/DomainTheory/1994} and defining evaluation semantics as the 
least fixed point of a suitable map, as it is customary in monadic evaluation semantics 
\cite{DalLagoGavazzoLevy/LICS/2017}.}).
% \begin{align*}
%     \eval{\return{\valone}} &= \mndun(\valone)
%     &
%     \eval{(\abs{\varone}{\termone})\valone} &= \eval{\subst{\termone}{\valone}{\varone}}
%     &
%     \eval{\seq{\termone}{\termtwo}} &= \bindsymbol{\eval{\subst{\termtwo}{-}{\varone}}} \eval{\termone}
%     &
%     \eval{\ith{\pair{\valone_1}{\valone_2}}} &= \mndun(\valone_i)
%     %&
%   % \eval{\snd{\pair{\valone}{\valtwo}}} &= \mndun(\valtwo)
%     &
%     \eval{\gop} &= \gopsem.
% \end{align*}
$$
\eval{\return{\valone}} = \mndun(\valone)
    \quad
    \eval{(\abs{\varone}{\termone})\valone} = \eval{\subst{\termone}{\valone}{\varone}}
     \quad
    \eval{\seq{\termone}{\termtwo}} = \kleisli{\eval{\subst{\termtwo}{\cdot}{\varone}}} \eval{\termone}
     \quad
    \eval{\ith{\pair{\valone_1}{\valone_2}}} = \mndun(\valone_i)
     \quad \hspace{-0.1cm}
    \eval{\gop} = \gopsem
$$    
Notice that according to the fine-grained methodology, evaluation is indeed defined only on computations 
(i.e. not on values).
\end{exa}

\cref{ex:algebraic-effects} defines evaluation semantics as an arrow in the category $\Set$ of 
sets and functions relying on two main ingredients: the monad $\mnd$ and its algebraic operations; and
primitive functions implementing the basic mechanism of $\beta$-reductions, viz. application/substitution and projections. 
The very same recipe has been used to define specific evaluation semantics beyond $\Set$, 
a prime example being kernel-like evaluation semantics for stochastic $\lambda$-calculi \cite{dal-lago/gavazzo-mfps-2019,DBLP:journals/pacmpl/WandCGC18,DBLP:journals/pacmpl/VakarKS19,EhrhardPT18,DBLP:journals/pacmpl/BartheCLG22}, 
where evaluation semantics are defined as Kleisli arrows on suitable categories of measurable spaces 
(see \cref{ex:measurable-semantics} below for details).
Here, we propose a general notion of an operational semantics for $\lambdacalc$ in an 
arbitrary cartesian category $\B$ and with respect to a monad $\mnd$. 
We call the resulting notion a \emph{($\syn$-)operational structure}. 

\begin{defi}
\label{def:evaluation-structure}
    Given a strong monad $\mnd$ 
    on a cartesian category $\B$, a \emph{($\syn$-)operational structure} consists of a diagram $\shallow: \syn \to \B$ 
    satisfying $\shallow(\vect{\varone: \typeone})  = \prod \vect{\shallow \typeone}$,
    together with the following (interaction) arrows (notice that $\gopsem$ ranges over generic effects) 
    and satisfying the coherence laws in \cref{figure:coherence-conditions-operational-structure}.
    \begin{align*}
    \ev &: \shallow\comptype{\typeone} \to \mndfun(\shallow \typeone)
    &
    \unitarrow &: 1 \to \shallow \unittype
    &
    \app &: \shallow(\typeone \to \typetwo) \times \shallow \typeone \to \shallow \comptype{\typetwo}
    & 
    \constarrow &: 1 \to \shallow \basetype 
    \\
    \projl &: \shallow(\typeone \times \typetwo) \to \shallow \typeone
    &
    \projr &: \shallow(\typeone \times \typetwo) \to \shallow \typetwo
    &
    \gopsem &: 1 \to \mndfun(\shallow \typeone_{\gop})
    \end{align*}
\end{defi}

\begin{figure*}[t]
%	\footnotesize
% \begin{tabular}{ c c c}
% {
%     \(
%      \xymatrix{
%     \shallow \envone \times \shallow \typeone 
%     \ar[r]^{\shallow \termone} 
%     \ar[d]_{\shallow(\lambda \varone.\termone) \times \id_{\shallow{\typeone}}}
%     & \shallow \comptype{\typetwo}
%     \\
%     \shallow(\typeone \to \typetwo) \times \shallow \typeone 
%     \ar[ru]_{\app}
%     }
%     \)
% }
% & 
% {
%     \(
%     \xymatrixcolsep{3pc}
%     \xymatrix{
%     \shallow \envone 
%     \ar[d]_-{\shallow \pair{\valone_1}{\valone_2}}
%     \ar[r]^{\shallow \valone_i}
%     & \shallow\typeone_i 
%     \\
%      \shallow(\typeone_1 \times \typeone_2)
%      \ar[ru]_{\bm{i}}
%     }
%     \)
% }
% & 
% {
%     \(
%     \xymatrix{
%     \shallow \envone 
%     \ar[r]^{\shallow k} 
%     \ar[d]_{!}
%     & \shallow \mathbf{k}
%     \\
%     1 
%     \ar[ru]_{\bm{\kappa}}
%     }
%     \)
% }
% % & 
% % {
% %     \(
% %     \xymatrix{
% %     \shallow \envone 
% %     \ar[r]^{\shallow \cbase} 
% %     \ar[d]_{!}
% %     & \shallow \basetype
% %     \\
% %     1 
% %     \ar[ru]_{\constarrow}
% %     }
% %     \)
% % }
% % \\
% % & & 
% % \\
% % \(\infer{\envone \valimp \lambda x.\termone: \typeone \to \typetwo}{\envone, x: \typeone \compimp \termone: \typetwo}\) 
% % & 
% % \(\infer{\envone \valimp \pair{\valone_1}{\valone_2}: \typeone_1 \times \typeone_2}{\envone \valimp \valone_1: \typeone_2 & 
% % \envone \valimp \valone_2:\typeone_2}
% % \)
% % &
% % \(\infer{\envone \valimp \unitval: \unittype}\)
% % \\
% \end{tabular}
$$
     \xymatrix{
    \shallow \envone \times \shallow \typeone 
    \ar[r]^{\shallow \termone} 
    \ar[d]_{\shallow(\lambda \varone.\termone) \times \id_{\shallow{\typeone}}}
    & \shallow \comptype{\typetwo}
    \\
    \shallow(\typeone \to \typetwo) \times \shallow \typeone 
    \ar[ru]_{\app}
    }
    \hspace{-0.4cm}
    \xymatrixcolsep{3pc}
    \xymatrix{
    \shallow \envone 
    \ar[d]_-{\shallow \pair{\valone_1}{\valone_2}}
    \ar[r]^{\shallow \valone_i}
    & \shallow\typeone_i 
    \\
     \shallow(\typeone_1 \times \typeone_2)
     \ar[ru]_{\bm{i}}
    }\
    \xymatrix{
    \shallow \envone 
    \ar[r]^{\shallow \unitval} 
    \ar[d]_{!}
    & \shallow \unittype
    \\
    1 
    \ar[ru]_{\unitarrow}
    }\
     \xymatrix{
     \shallow \envone 
     \ar[r]^{\shallow \cbase} 
     \ar[d]_{!}
     & \shallow \basetype
     \\
     1 
     \ar[ru]_{\constarrow}
    }
$$
% 	\caption{Coherence condition evaluation structure, where $\bm{i} \in \{\projl, \projr\}$.}
% 	\label{figure:coherence-conditions-evaluation-structure}
% \end{figure*}

% A (monadic) evaluation semantics is then defined as a $\B$ arrow satisfying suitable coherence conditions. 

% \begin{defi}
% \label{def:evaluation-semantics}
%     Given an evaluation structure as in \cref{def:evaluation-structure} and a strong monad $\mnd = (\mndfun, \mndun, \bindsymbol)$ 
%     on $\B$ with generic effects $\sem{\gop}_{\mndfun}$ for any generic effect symbol $\gop$, 
%     an evaluation semantics is family of arrows $\ev: \shallow\comptype{\typeone} \to \mndfun(\shallow \typeone)$ 
%     satisfying the coherence conditions in 
%     \cref{figure:coherence-conditions-evaluation-semantics}.
% \end{defi}

% \begin{figure*}[t]
% %	\footnotesize
\begin{tabular}{ c c c}
{
    \(
    \xymatrixcolsep{3pc}
    \xymatrix{
    \shallow \envone 
    \ar[r]^{\shallow(\return{\valone})}
    \ar[d]_{\shallow(\valone)}
    & \shallow \comptype{\typeone} 
    \ar[d]^{\ev}
    \\
    \shallow\typeone 
    \ar[r]_{\mndun} 
    & \mndfun(\shallow \typeone)
    }
    \)
}
& 
{
\xymatrix{
    \shallow \envone \ar[d]_{!} \ar[r]^-{\shallow(\gop)} & \shallow(\comptype{\typeone}) \ar[d]^{\ev} 
    \\
    1 \ar[r]_-{\gopsem} & \mndfun(\shallow \typeone)}
}
&
{
    \(
    \xymatrixcolsep{2pc}
    \xymatrix{
    \shallow \envone 
    \ar[rr]^-{\shallow(\valone\valtwo)}
    \ar[d]_{\lan\shallow(\valone), \shallow(\valtwo)\ran}
    &
    & \shallow \comptype{\typetwo} \ar[d]^{\ev}
    \\
    \shallow(\typeone \to \typetwo) \times \shallow\typeone 
    \ar[r]_-{\app} 
    &  \shallow \comptype{\typetwo} 
    \ar[r]_-{\ev} 
    & \mndfun(\shallow \typetwo)
    }
    \)
}
% \\
% & & 
% \\
% \(\infer{\envone \compimp \return{\valone}: \typeone}{\envone \valimp \valone: \typeone}\) 
% & 
% {
% \(\infer{\envone \compimp \valone\valtwo: \typetwo}{\envone \valimp \valone: \typeone \to \typetwo & 
% \envone \valimp \valtwo:\typeone}
% \)
% }
\end{tabular}
\begin{tabular}{ c c}
{
    \(
    \xymatrixcolsep{2.9pc}
    \xymatrix{
    \shallow \envone 
    \ar[rr]^-{\shallow(\seq{\termone}{\termtwo)}}
    \ar[d]_{\lan\id, \shallow(\termone)\ran}
    &
    &\shallow \comptype{\typetwo} \ar[d]^{\ev}
    \\
    \shallow \envone \times \shallow \comptype{\typeone} 
    \ar[r]_-{\id \times \ev} 
    &  \shallow\envone \times \mndfun(\shallow \typeone) 
    \ar[r]_-{{\footnotesize\bindsymbol}(\ev \circ \shallow(\termtwo))} 
    & \mndfun(\shallow \typetwo)
    }
    \)
}
&
{
    \(
    \xymatrix{
    \shallow \envone 
    \ar[rr]^-{\shallow(\ith{\valone})}
    \ar[d]_{\shallow(\valone_i)}
    &
    &\shallow \comptype{\typeone_i} \ar[d]^{\ev}
    \\
    \shallow(\typeone_1 \times \typeone_2)
    \ar[r]_-{\bm{i}} 
    &  \shallow \typeone_i 
    \ar[r]_-{\mndun} 
    & \mndfun(\shallow \typeone_i)
    }
    \)
}
\end{tabular}
	\caption{Coherence laws, where $\ith{-} \in \{\fst{-}, \snd{-}\}$ and $\bm{i} \in \{\projl, \projr\}$.
	%$k \in \{\unitval, \cbase\}$, $\bm{\kappa} \in \{\unitarrow,\constarrow\}$, $\mathbf{k} \in \{\unittype, \basetype\}$
	}
%	and $\mathsf{ith} \in \{\projl,\projr\}$.}
%		\caption{Coherence condition evaluation structure, where $\bm{i} \in \{\projl, \projr\}$.}
	\label{figure:coherence-conditions-operational-structure}
\end{figure*}

Notice how the first four coherence laws in \cref{figure:coherence-conditions-operational-structure} 
ensure the intended behaviour of the arrows $\app, \unitarrow, \bm{p_i},\constarrow$, whereas the remaining laws  
abstractly describe the main dynamics of program execution. Notice also that \cref{def:evaluation-structure} 
prescribes the existence of an evaluation arrow $\ev$: it would be interesting to find conditions on $\B$ 
(probably a domain-like enrichment \cite{Kelly/EnrichedCats} or partial additivity \cite{manesArbib86}) 
ensuring the existence of $\ev$. 
We can now instantiate \cref{def:evaluation-structure}
to recover standard $\Set$-based evaluation semantics as well as operational semantics on richer categories. 
In particular, since $\B$ need not be closed, we can give 
$\lambdacalc$ an operational semantics in the category $\Meas$ of measurable spaces and measurable functions.
%(\cref{ex:measurable-semantics}).

\begin{exa}%[$\lambda$-calculus with algebraic operations]
Let $\B = \Set$ and define $\shallow: \syn \to \Set$ thus:
% \begin{align*}
%     \shallow \typeone &= \values_{\typeone}
%     &
%     \shallow \comptype{\typeone} &= \terms_{\typeone}
%     &
%     \shallow(\varone_1: \typeone_1, \hdots, \varone_n: \typeone_n) &= \shallow \typeone_1 \times \cdots \times \shallow \typeone_n
% \end{align*}
% \vspace{-1cm}
% \begin{align*}
% \shallow(\envone \compimp \termone: \typeone)(\vect{\valone}) &= \subst{\termone}{\vect{\valone}}{\vect{\varone}}
% &
% \shallow(\envone \valimp \valtwo: \typeone)(\vect{\valone}) = \subst{\termone}{\vect{\valone}}{\vect{\varone}}
% \end{align*}
\begin{align*}
    \shallow \typeone &= \values_{\typeone}
    &
    \shallow \comptype{\typeone} &= \terms_{\typeone}
    &
    \shallow(\vect{\varone: \typeone}) &= \prod \vect{\shallow\typeone}
    &
\shallow(\envone \compimp \termone: \typeone)(\vect{\valone}) &= \subst{\termone}{\vect{\valone}}{\vect{\varone}}
\end{align*}
%\quad
%\shallow(\envone \valimp \valtwo: \typeone)(\vect{\valone}) = \subst{\termone}{\vect{\valone}}{\vect{\varone}}
We obtain an operational structure by defining the maps $\app, \projl, \projr$, and $\unitarrow$ 
in the obvious way (e.g., $\app(\abs\varone\termone, \valone) = \subst\termone\valone\varone$ and 
$\bm{p_i}\pair{\valone_1}{\valone_2} = \valone_i$). 
Finally, given a monad $\mnd$ with generic effects
$\gopsem \in \mndfun \values_{\typeone_{\gop}}$ for each effect symbol $\gop$, 
we define $\ev$ as the evaluation map of \cref{ex:algebraic-effects}.
\end{exa}

\begin{exa}[Stochastic $\lambda$-calculus]
\label{ex:measurable-semantics}
Let us consider the instance of $\lambdacalc$ with a base type $\mathbf{R}$ for real numbers, 
constants $c_{r}$ for each real number $r$, and the generic effect symbol $\mathcal{U}$ 
standing for the uniform distribution over the unit interval.
% The static semantics of $\lambdacalc$ is extended thus:
% \[
% \infer{\Gamma \valimp c_r: \mathbf{Real}}{\phantom{x}}
% \qquad
% \infer{\Gamma \compimp c_f(\valone_1, \hh, \valone): \mathbf{Real}}
% {\Gamma \valimp \valone_1: \mathbf{Real} & \cdots & \Gamma \valimp \valone_n: \mathbf{Real}}
% \qquad
% \infer{\envone \compimp \mathbf{sample}: \mathbf{Real}}{\phantom{\varone}}
% \]
%Let us now move to operational semantics.
Recall that $\Meas$ has countable products and coproducts
(but not exponentials \cite{aumann}). To define the diagram $\shallow: \syn \to \Meas$, we 
rely on the well-known fact \cite{EhrhardPT18,StatonYWHK16,DBLP:conf/icfp/BorgstromLGS16} 
that both $\values_{\typeone}$ and $\terms_{\typeone}$ 
can be endowed with a $\sigma$-algebra making them measurable (actually Borel) spaces 
in such a way that $\values_{\mathbf{R}} \cong \mathbb{R}$ and that
the substitution map is measurable. We write $\Sigma_{\typeone}$ and 
$\Sigma_{\comptype{\typeone}}$ for the $\sigma$-algebras associated to 
$\values_{\typeone}$ and $\terms_{\typeone}$, respectively.
We thus define $\shallow: \syn \to \Meas$ as follows: 
\begin{align*}
    \shallow \typeone &= (\values_{\typeone}, \Sigma_{\typeone}) 
    &
    \shallow \comptype{\typeone} &= (\terms_{\typeone}, \Sigma_{\comptype{\typeone}})
    %&
    %\shallow \envone &= \shallow \typeone_1 \otimes \cdots \otimes \shallow \typeone_n 
    &
    \shallow(\vect{\varone: \typeone}) &= \prod \vect{\shallow\typeone}
    &
\shallow(\envone \compimp \termone: \typeone)(\vect{\valone}) &= \subst{\termone}{\vect{\valone}}{\vect{\varone}}.
%&
%\shallow(\envone \valimp \valtwo: \typeone)(\vect{\valone}) = \subst{\termone}{\vect{\valone}}{\vect{\varone}} 
\end{align*}
We obtain an operational structure by observing that the maps $\app, \projl, \projr$, and $\unitarrow$ 
of the previous example extend to $\Meas$, in the sense that they are all measurable functions. 
Next, we consider the Giry monad \cite{10.1007/BFb0092872} $\mathbb{G} = (G, \eta, \mu)$ 
with
$G: \Meas \to \Meas$ mapping each measurable space to the space of 
probability measures on it. 
%The unit of $\giry$
%associates to each element $x$ the Dirac distribution $\dirac{x}$, 
%and for a measurable function $f: X \to \giry(Y)$, the map
%$\bindsymbol{f}$ defined by 
%$\bindsymbol{f}(\mu)(B) = \int f(x)(B) \mu(dx)$ 
%gives the Kleisli extension of $f$ 
By Fubini-Tonelli theorem, $\mathbb{G}$ is strong. 
Let $\bm{\mathcal{U}}$ be the Lebesgue measure on $[0,1]$, which we regard as an arrow
 $\bm{\mathcal{U}}: 1 \to \giry(\shallow \values_{\mathbf{Real}})$, and thus as a generic effect in $\Meas$.
We then define \cite{dal-lago/gavazzo-mfps-2019}
%(also known as a probability kernel \cite{DBLP:journals/entcs/Panangaden99})
$\ev: \terms_{\typeone} \to \giry(\values_{\typeone})$ 
as in \cref{ex:algebraic-effects}. 
% 	\begin{align*}
% 		\sem{\valone}(A) 
% 		&= \dirac{\valone}(A)
% 		\\
% 	%	\sem{c_f(c_{r_1}, \hh, c_{r_n})}(A)
% 	%	&=
% 	%	\sem{c_{f(r_1, \hh, r_n)}}(A)
% 	%	\\
% 		\sem{\mathbf{sample}}(A)
% 		&= \sem{\mathbf{sample}}_{\giry} \{r \mid c_{r} \in A\}
% 		\\
% 		\sem{(\abs{\varone}{\termone})\valone}(A) 
% 		&= \sem{\subst{\termone}{\varone}{\valone}}(A)
% 		\\
% 		\sem{\seq{\termone}{\termtwo}}(A)
% 		&=
% 		 \int \sem{\subst{\termtwo}{\valone}{\varone}}(A) \sem{\termone}(d \valone)
% 	\end{align*} 
% 	\begin{align*}
% 		\sem{\valone}
% 		&= \mndun(\valone)
% 		\\
% 		\eval{\fst{\pair{\valone}{\valtwo}}} &=  \mndun(\valone)
%         \\
%         \eval{\snd{\pair{\valone}{\valtwo}}} &=  \mndun(\valtwo)
%         \\
% 		\sem{\mathbf{sample}} &= \sem{\mathbf{sample}}_{\giry}
% 		\\
% 		\sem{(\abs{\varone}{\termone})\valone} &= \sem{\subst{\termone}{\varone}{\valone}}
% 		\\
% 		\sem{\seq{\termone}{\termtwo}}
% 		&=
% 		 \bindsymbol{\eval{\subst{\termtwo}{-}{\varone}}} \eval{\termone}
% 	\end{align*} 
% (notice that $\sem{-}$ is well-defined, since $\lambdacalc$ is simply typed).
Alternatively, we can define a finer operational structure by exploiting the 
observation that  $\values_{\typeone}$ and $\terms_{\typeone}$ 
are Borel spaces and move from $\Meas$ to 
its full subcategory $\Pol$ whose objects are standard Borel spaces (notice that arrows in $\Pol$ 
are measurable functions).
Indeed, all the construction seen so far actually gives standard Borel spaces. 
In particular, the Giry monad gives a monad on $\Pol$ \cite{kechris1995classical}
and the operational structure thus obtained coincides with the stochastic  
operational semantics in \cite{DBLP:journals/pacmpl/BartheCLG22}.
\end{exa}

%% file: prelim.tex
\section{Operational Logical Relations, Fibrationally}
\label{sect:prelim}

Having defined what an operational semantics for $\lambdacalc$ is, we now focus on \emph{operational} 
reasoning. In this section, we propose a general notion of 
an operational logical relation in terms of fibrations over (the underlying category of) an operational
structure and prove that a general version of the fundamental lemma of logical relations 
holds for our operational logical relations. But before that, let us recall some 
preliminary notions on (bi)fibrations (we refer to \cite{Hermida93,Jacobs01,Streicher18} for more details).

\subsection{Preliminaries on Fibrations}
Let \fun{\fib}{\E}{\B} be a functor and \fun{f}{X}{Y} an arrow in \E with $\fib(f) = u$. 
We say that $f$ is \emph{cartesian} over $u$ if,  
for every arrow \fun{h}{Z}{Y} in \E such that $\fib(h) = u\circ v$, there is a unique arrow \fun{g}{Z}{X} such that $\fib(g) = v$ and $ h = f\circ g$. 
Dually, $f$ is \emph{cocartesian} over $u$ if, 
for every arrow \fun{h}{X}{Z} in \E such that $\fib(h) = v\circ u$, there is a unique arrow \fun{g}{Y}{Z} such that $\fib(g) = v$ and $ h = g\circ f$. 
We say that $f$ is \emph{vertical} if $u$ is an identity. 

A \emph{fibration} is a functor \fun{\fib}{\E}{B} such that, 
for every object $X$ in \E and every arrow \fun{u}{I}{\fib(X)} in \B,  the exists a cartesian arrow over $u$ with codomain $X$. 
Dually, an \emph{opfibration} is a functor \fun{\fib}{\E}{\B}  such that 
for every object $X$ in \E and every arrow \fun{f}{\fib(X)}{I} in \B, the exists a cocartesian arrow over $u$ with domain $X$. 
A \emph{bifibration} is a functor which is both a fibration and an opfibration. 
We refer to $\E$ and $\B$ as the domain and the base of the (bi/op)fibration. 
A (op)fibration is \emph{cloven} if it comes together with a choice of (co)cartesian liftings: 
for an object $X$ in \E, we denote by 
\fun{\liftar u X}{\lift u X}{X} the chosen cartesian arrow over \fun{u}{I}{\fib(X)} and by 
\fun{\coliftar u X}{X}{\colift u X} the chosen cocartesian arrow over \fun{u}{\fib(X)}{I}. 
A bifibration is cloven if it has choices both for cartesian and cocartesian liftings. 
From now on, we assume all (bi/op)fibrations to be cloven. 

Let \fun{\fib}{\E}{\B} be a functor and $I$ an object in \B. 
The \emph{fibre} over $I$ is the category $\fibre\E{I}$ where objects are objects $X$ in \E such that $\fib(X) = I$ and 
arrows are arrows \fun{f}{X}{Y} in \E such that $\fib(f) = \id_I$, namely, vertical arrows over $I$. 
Then, for every arrow \fun{u}{I}{J}, the following hold:
\begin{itemize}
\item if $\fib$ is a fibration,  we have a functor 
\fun{\lift u}{\fibre\E J}{\fibre\E I} called \emph{reindexing along $u$}; 
\item if $\fib$ is an opfibration, we have a functor 
\fun{\colift u}{\fibre\E I}{\fibre\E J} called \emph{image along $u$}; 
\item if $\fib$ is a bifibration, we have an adjunction $\colift u \dashv \lift u$. 
\end{itemize}  

\begin{exa}%[$n$-ary predicates] 
\label{ex:subsets}  \label{ex:rel}
Let us consider the category 
$\PredC\Set$ of homogeneous $n$-ary predicates over sets. 
Its objects are  pairs of sets \ple{X,A} with $A\subseteq X^n$ as objects and 
arrows \fun{f}{\ple{X,A}}{\ple{Y,B}} are functions 
\fun{f}{X}{Y} such that $A\subseteq (\prod f)^{-1}(B)$. 
There is an obvious functor 
\fun{\Pred[\Set]}{\PredC\Set}{\Set}
mapping an $n$-ary predicate \ple{X,A} to its underlying set $X$ and an arrow \fun{f}{\ple{X,A}}{\ple{Y,B}} to its underlying function \fun{f}{X}{Y}. 
The functor $\Pred[\Set]$ is a bifibration where 
the cartesian and cocartesian liftings of a function \fun{f}{X}{Y}  are given by inverse and direct images along $\prod f$, respectively.
Special instances of $\Pred[\Set]$ are obtained for $n=1$ (unary predicates) and $n=2$ (binary relations). 
In those cases, we specialise the notation and write \fun{\Sub[\Set]}{\SubC\Set}{\Set} and \fun{\Rel\Set}{\RelC\Set}{\Set}. 
\end{exa} 

\begin{exa}
Starting from \fun{\Rel\Set}{\RelC\Set}{\Set}, we obtain a (bi)fibration 
of relations $\mathsf{Rel}_{\Meas}: \mathsf{r}(\Meas) \to \Meas$ on measurable spaces 
simply by considering the category $\mathsf{r}(\Meas)$ --- whose objects are triples $(X, \Sigma, R)$, 
with $(X, \Sigma)$ a measurable space and $R$ a relation on $X$ --- 
and forgetting about the measurable space structure, hence mapping 
$(X, \Sigma, R)$ to $(X, R)$.
\end{exa}

\begin{exa}[Weak subobjects]\label{ex:wsub}
Let \CC be a category with weak pullbacks. 
We define the bifibration \fun{\wSub[\CC]}{\WS\CC}{\CC} of weak subobjects in \CC \cite{Grandis00,MaiettiR13}. 
Let $I$ be an object of \CC and denote by $\CC/I$ the slice over $I$, that is, 
objects of $\CC/I$ are arrows \fun{\alpha}{X}{I} in \CC and an arrow \fun{f}{\alpha}{\beta}, where \fun{\alpha}{X}{I} and \fun{\beta}{Y}{I} are objects of $\CC/I$, is an arrow \fun{f}{X}{Y} such that $\alpha = \beta\circ f$. 
For $\alpha,\beta$ objects in $\CC/I$, we write $\alpha \leq \beta$ if there exists an arrow \fun{f}{\alpha}{\beta} in $\CC/I$. 
This is a preorder on the objects of $\CC/I$ and, like any preorder, induces an equivalence on the objects of $\CC/I$ defined as follows: 
$\alpha\equiv \beta$ if and only if $\alpha\leq\beta$ and $\beta\leq\alpha$, that is, if and only if there are arrows \fun{f}{\alpha}{\beta} and \fun{g}{\beta}{\alpha} in $\CC/I$. 
Then, the category $\WS\CC$ is given as follows. 
Objects of  $\WS\CC$ are pairs \ple{I,[\alpha]} where $I$ is an object of \CC and $[\alpha]$ is an equivalence class of objects of $\CC/I$ modulo $\equiv$ 
and an arrow \fun{f}{\ple{I,[\alpha]}}{\ple{J,[\beta]}} is an arrow \fun{f}{I}{J} in \CC such that 
$f\circ\alpha\leq\beta$ or, more explicitly, 
$f\circ \alpha = \beta\circ g$ for some arrow $g$ in \CC. 
It is easy to check that this definition does not depend on the choice of representatives. 
Finally, composition and identities in $\WS\CC$ are those of \CC. 
The functor $\wSub[\CC]$ maps \ple{I,[\alpha]} to $I$ and is the identity on arrows. 
It is easy to check that, 
an arrow \fun{f}{\ple{I,[\alpha]}}{\ple{J,[\beta]}} is 
vertical iff $f = \id_I$, 
cocartesian iff $[\beta] = [f\circ\alpha]$ and 
cartesian iff $[\alpha] = [\beta']$ where $f\circ\beta' = \beta\circ f'$ is a weak pullback square in \CC. 
Therefore, we have that, 
for every arrow \fun{f}{I}{J} in \CC and every object \ple{I,[\alpha]},
the image along $f$ is $\colift f {\ple{I,[\alpha]}} = \ple{J,[f\circ\alpha]}$, and, 
for every object \ple{J,[\beta]}, the reindexing along $f$ is $\lift f {\ple{J,[\beta]}} = \ple{I,[\beta']}$, where $f\circ\beta' = \beta\circ f'$ is a weak pullback square in \CC.
Therefore, $\wSub[\CC]$ is a bifibration.\footnote{This is actually an instance of the Gr\"{o}thendieck construction applied to the doctrine of weak subobjects as described in \cite{MaiettiR13}.}
Finally, observe that the bifibrations $\wSub[\Set]$ and $\Sub[\Set]$ are equivalent in the sense that there is a functor 
\fun{U}{\Sub[\Set]}{\WS\Set} which is an equivalence satisfying $\wSub[\Set] \circ U = \Sub[\Set]$ and preserving (co)cocartesian arrows. 
The functor $U$ maps \ple{X,A} to \ple{X,[\iota_A]} where \fun{\iota_A}{A}{X} is the inclusion function, and it is the identity on arrows. 
\end{exa}

% \begin{exa}
% Following the previous example, we obtain the fibrations of 
% \emph{Borel equivalence relations}~\cite{kechris1995classical}. 
% Let us recall that for a standard Borel space $(X, \Sigma)$, a 
% Borel equivalence relation is an equivalence relation on $X$ such that 
% its graph is Borel (that is, it is a subobject of the product space  
% $(X, \Sigma) \otimes (X, \Sigma)$ that is reflexive, symmetric, and transitive). 
% We thus obtain a category whose objects are triples $(X, \Sigma, R)$, with 
% $(X,\Sigma)$ object in $\Pol$ and $R$ Borel equivalence relation on $X$, 
% and whose arrows are measurable functions preserving relations. 
% There is an obvious functor from such a category to 
% $\Pol$ mapping $(X, \Sigma, R)$ to $(X,\Sigma)$ which gives raise 
% to a fibration.
% % Such a theory has been instantiated to study bisimulation and logic-based equivalences, 
% % as well as contextual equivalence. Here, we show how logical relations can be obtained 
% % as a direct instance of our framework. 
% \end{exa}

\begin{exa}[Information Flow]
We now define a (bi)fibration of classified relations~\cite{DBLP:journals/pacmpl/Kavvos19} in terms of 
presheaves. 
Fixed a lattice $\mathcal{L}$ of security levels (such as the two-elements 
lattice $\{\mathtt{private} \leq \mathtt{public}\}$) regarded as a category, we consider the 
functor category 
$\Set^{\mathcal{L}}$ and the category $\mathsf{r}(\Set^{\mathcal{L}})$ whose objects are 
pairs $(X, R)$ with $X: \mathcal{L} \to \Set$ and $R$ mapping each $\ell \in \mathcal{L}$ 
to a binary endorelation $R_\ell$ on $X(\ell)$ such that $\ell_1 \leq \ell_2$ implies 
$R_{\ell_1} \subseteq R_{\ell_2}$, and whose arrows are natural transformations preserving 
relations. As usual, there is an obvious forgetful functor from 
 $\mathsf{r}(\Set^{\mathcal{L}})$ to $\Set^{\mathcal{L}}$ mapping $(X, R)$ to $X$, 
 which gives rise to a fibration. 
Such a fibration naturally supports relational reasoning about information flow as follows~\cite{DBLP:conf/popl/AbadiBHR99,DalLagoG22}. 
Starting with the $\Set$-based operational structure of $\lambda$-terms, 
we embed such a structure in $\Set^{\mathcal{L}}$ by considering the constant functor 
mapping each $\ell$ to the (constant) set of $\lambda$-terms. Moving to 
$\mathsf{r}(\Set^{\mathcal{L}})$, we then obtain relations on $\lambda$-terms that 
depend on security levels, and we think about them as giving program equivalence 
for observers with security levels (or permissions) in $\mathcal{L}$.
\end{exa}

\begin{exa}[Simple fibration] \label{ex:simp-fib}
Let \CC be a category with finite products. 
We consider a category $\simp\CC$ defined as follows. 
Objects are pairs \ple{I,X} of objects of \CC and an arrow \fun{\ple{u,f}}{\ple{I,X}}{\ple{J,Y}} consists of a pair of arrows \fun{u}{I}{J} and \fun{f}{I\times X}{Y} in \CC. 
The composition of \fun{\ple{u,g}}{\ple{I,X}}{\ple{J,Y}} and \fun{\ple{v,g}}{\ple{J,Y}}{\ple{K,Z}} in $\simp\CC$  is given by 
$\ple{v,g}\circ\ple{u,f} = \ple{v\circ u, g\circ\iple{u\circ\pi_1,f}}$, where \fun{\pi_1}{I\times X}{I} is the first projection and \fun{\iple{u\circ\pi_1,f}}{I\times X}{J\times Y} is the pairing of $u\circ\pi_1$ and $f$ given by the universal property of the product $J\times Y$. 
The identity on \ple{I,X} is the arrow \fun{\ple{\id_X,\pi_2}}{\ple{I,X}}{\ple{I,X}}, where \fun{\pi_2}{I\times X}{X} is the second projection. 
There is an obvious functor \fun{\sfib\CC}{\simp\CC}{\CC}, which maps an object \ple{I,X} to the first component $I$ and acts similarly on arrows. 
This functor is a fibration known as the \emph{simple fibration} \cite{Jacobs01}, where, 
for every arrow \fun{u}{I}{J} in \CC and object \ple{J,X} in $\simp\CC$, the cartesian lifting of $u$ is the arrow 
\fun{\liftar{u}{\ple{J,X}} = \ple{u,\pi_2}}{\ple{I,X}}{\ple{J,X}}. 
\end{exa}

Fibrations nicely carry a logical content: logical operations, in fact, can be described as categorical structures on the fibration. 
We now define the logical structure underlying logical relations, namely conjunctions, implications, and universal quantifiers. 
Recall that a fibration \fun\fib\E\B\ has \emph{finite products} if 
\E and \B have finite products and $\fib$ preserves them. 
We denote by $\dtimes$ and $\dter$ finite products in \E and recall  
that in a fibration \fun\fib\E\B with finite products every fibre $\fibre\E{I}$ has finite products $\land_I$ and $\top_I$\footnote{We will often omit subscripts when they are clear from the context.} preserved by reindexing functors.
Additionally, we have the isomorphisms 
$A\dtimes B \cong \lift{\pi_1}{A} \land_{\fib(A)\times\fib(B)} \lift{\pi_2}{B}$ and $\dter \cong \top_\terobj$. 
This follows essentially because $\dtimes$, being a right adjoint, preserves cartesian arrows (\cf~\cite[Exercise 1.8.5]{Jacobs01}). 

\begin{defi}\label{def:log-fib}
A fibration \fun\fib\E\B\ with finite products is a \emph{logical fibration} if it is fibred cartesian closed and has universal quantifiers, 
where:
\begin{enumerate}
\item $\fib$ is \emph{fibred cartesian closed} if every fibre $\fibre\E I$ has exponentials, denoted by $X\impl Y$, 
and reindexing preserves them.
\item $\fib$ has \emph{universal quantifiers} if, for every projection \fun{\pi}{I\times J}{I} in \B, the reindexing functor \fun{\lift{\pi}}{\fibre\E I}{\fibre\E{I\times J}} has a right adjoint \fun{\uq I J}{\fibre\E{I\times J}}{\fibre\E I} satisfying the Beck-Chevalley condition \cite{Jacobs01}. 
\end{enumerate}
\end{defi}

Note that in a fibration with universal quantifiers, we have right adjoints along any tuple of distinct projections 
\fun{\iple{\pi_{i_1},\ldots,\pi_{i_k}}}{I_1\times\ldots\times I_n}{I_{i_1}\times\ldots\times I_{i_k}},
where $i_1,\ldots,i_k \in \{1, \hdots, n\}$ are all distinct. 
We denote such a right adjoint by $\uqn{\iple{\pi_{i_1},\ldots,\pi_{i_k}}}$. 

The following propositions %adapts results from \cite{Jacobs01,MaiettiPR21} and 
shows under which conditions the fibration of weak subobjects and the simple fibration are logical fibrations.

\begin{propC}[{\cite[Prop. 7.6]{MaiettiPR21}}]
\label{prop:wsub-lf}
Let \CC be a category with finite products and weak pullbacks. 
% \item\label{prop:wsub-lf:1} If \CC is cartesian closed, then $\wSub[\CC]$ has universal quantifiers; 
% \item\label{prop:wsub-lf:2} 
If \CC is slicewise weakly cartesian closed, then the fibration of weak subobjects $\wSub[\CC]$ is a logical fibration. 
\end{propC}

\begin{propC}[{\cite[Exercise 1.3.2, Prop. 1.9.3]{Jacobs01}}]
\label{prop:simp-lf} 
If \CC is cartesian closed, then 
the simple fibration $\sfib\CC$ is a logical fibration. 
\end{propC}

\begin{exa}
Since \Set is locally cartesian closed,  \cref{prop:wsub-lf} implies that both $\wSub[\Set]$ and $\Sub[\Set]$, being equivalent as noticed in \cref{ex:wsub},
are logical fibrations. 
Also $\Rel\Set$ is a logical fibration, as it can be obtained from $\Sub[\Set]$ by pulling back along the product-preserving 
functor \fun{X \mapsto X \times X}{\Set}{\Set}.
\end{exa}

\subsubsection*{A 2-category of (bi)fibrations} 
Fibrations can be organised in a 2-category. 
This is important because many standard categorical concepts can be internalised in any 2-category, thus providing us with an easy way to define them also for fibrations. 
In particular, we will be interested in (strong) monads on a fibration, as they will allow us to define effectful logical relations. 

We consider the 2-category \Fib of fibrations defined as follows. 
Objects are fibrations \fun\fib\E\B. 
A 1-arrow \oneAr{F}{\fib}{\fibq} between fibrations \fun\fib\E\B\ and \fun\fibq\D\CC is a pair of functors \ple{\fb{F},\ft{F}} where 
\fun{\fb{F}}{\B}{\CC} and \fun{\ft{F}}{\E}{\D} and $\fb{F}\circ \fib = \fibq\circ \ft{F}$.\footnote{Note that we do not require $\ft{F}$ to preserve cartesian arrows.} 
A 2-arrow \twoAr{\phi}{F}{G}  between 1-arrows 
\oneAr{F,G}{\fib}{\fibq} is a pair of natural transformations \ple{\fb\phi,\ft\phi} where 
\nt{\fb\phi}{\fb{F}}{\fb{G}} and \nt{\ft\phi}{\ft{F}}{\ft{G}} and 
$\fb\phi\fib = \fibq\ft\phi$.
Compositions and identities are defined componentwise. 
The 2-category \biFib of bifibrations is the full 2-subcategory of \Fib whose objects are bifibrations. 
A 1-arrow \oneAr{F}{\fib}{\fibq} in \biFib is said to be cartesian if $\ft{F}$ preserves cartesian arrows; 
it is said to be cocartesian if $\ft{F}$ preserves cocartesian arrows. 

Following \cite{Street72}, we can define (strong) monads on fibrations as (strong) monads in the 2-category \Fib. 
That is, a \emph{monad} $\mnd$ on a fibration \fun\fib\E\B\ consists of the following data: 
\begin{itemize}
\item a 1-arrow \oneAr{T}{\fib}{\fib} and 
\item two 2-arrows \twoAr{\mu}{T^2}{T} and \twoAr{\eta}{\Id}{T} and 
\item such that $\fb\mnd = \ple{\fb{T},\fb\mu,\fb\eta}$ is a monad on \B and $\ft\mnd = \ple{\ft{T},\ft\mu,\ft\eta}$ is a monad on \E. 
\end{itemize} 
The monad $\ft\mnd$ is called a \emph{lifting} of the monad $\fb\mnd$ along the fibration $\fib$. 
If \fun\fib\E\B\  has finite products, a \emph{strong monad} on $\fib$ is a pair \ple{\mnd,\str}  where 
$\mnd = \ple{T,\mu,\eta}$ is a monad on $\fib$ and 
$\twoAr{\ple{\str}}{\ple{\blank\times \fb{T}\blank, \blank\dtimes \ft{T}\blank}}{\ple{\fb{T}(\blank\times\blank),\ft{T}(\blank\dtimes\blank)}}$ is a 2-arrow such that \ple{\fb\mnd,\fb\str} is a strong monad on \B and \ple{\ft\mnd,\ft\str} is a strong monad on \E. 

\begin{exa}
Given a monad $\mathbb{T} = (T, \eta, \mu)$ on $\Set$, we obtain a lifting 
of $\mathbb{T}$ along \fun{\Rel\Set}{\RelC\Set}{\Set}
via the so-called \emph{Barr extension} \cite{Barr/LMM/1970} of $\mathbb{T}$. 
Looking at a relation $R \subseteq X \times Y$ as the span  
$\xymatrix{X & \ar[l]_{\pi_1} R \ar[r]^{\pi_2} & Y,}$ we define the 
relation $\hat{T}R \subseteq TX \times TY$ as 
the span 
$\xymatrix{TX & \ar[l]_{T\pi_1} TR \ar[r]^{T\pi_2} & TY.}$
When the functor $T$ preserves weak pullbacks, then $\hat{T}$ extends $T$ 
to  an ordered  functor on $\ct{Rel}$, the ordered category of sets and relations,
making $\eta$ and $\mu$ lax natural transformations~\cite{Barr/LMM/1970}. 
Consequently, mapping objects $(X, R)$ of $\mathsf{r}(\Set)$ to 
$(TX, \hat{T}R)$ and arrows $f: (X,R) \to (Y,S)$ to 
$Tf: (TX, \hat{T}R) \to (TY, \hat{T}S)$, we obtain 
the desired lifting.
 Note that preservation of weak pullbacks is only needed to force $\hat{T}$ to strictly preserve relational composition. 
Thus, since we do not consider composition, restricting ourselves to endorelations, we can avoid this assumption. 
\end{exa}

\begin{exa}
Let us consider the Giry monad $\mathbb{G} = (G, \eta, \mu)$ and 
the fibration $\mathsf{Rel}_{\Meas}: \mathsf{r}(\Meas) \to \Meas$. 
Since relations in $\mathsf{r}(\Meas)$ are defined as set-theoretic relations, 
we can consider the Barr extension of the Giry monad as in the previous example. 
Unfortunately, that does not work in general~\cite{DBLP:journals/lmcs/KatsumataSU18,DBLP:journals/iandc/Katsumata13}, 
although it does when we restrict objects of $\mathsf{r}(\Meas)$ to structure
 $(X, \Sigma, R)$ such that $R$ is \emph{reflexive}. 
This obviously gives a fibration, and the Barr extension of $\mathbb{G}$ 
provides the desired lifting as in the previous example~\cite{dal-lago/gavazzo-mfps-2019}.
\end{exa}

\begin{exa}
Let us consider the fibration $\mathsf{Rel}_{\Set^{\mathcal{L}}}: \mathsf{r}(\Set^{\mathcal{L}}) \to 
\Set^{\mathcal{L}}$ of information flow. To model information hiding, we consider 
the identity monad $\mathbb{ID}$ on $\Set^{\mathcal{L}}$. Intuitively, information hiding does not 
directly act on programs, but rather on the way an external observer can inspect them. 
Consequently, the actual information-hiding operation will be modelled by the lifting of 
$\mathbb{ID}$ along $\mathsf{Rel}_{\Set^{\mathcal{L}}}$, which ultimately 
means specifying how $\mathbb{ID}$ acts on relations. 
For simplicity, let us consider the lattice $\mathtt{private} \leq \mathtt{public}$. 
Thinking about a relation $R(\ell)$ as program equivalence for an observer with security permission 
$\ell$, we define the masking operation as obscuring private programs, that is making them inaccessible to 
an external observer, 
so that private programs will be
regarded as equivalent.
Formally, define $!R(\mathtt{public}) = R(\mathtt{public})$ and 
$!R(\mathtt{private}) = X(\mathtt{private}) \times X(\mathtt{private})$. The assignment 
 $(X, R)$ to $(X, !R)$ then gives the desired lifting.
\end{exa}

\subsection{Operational Logical Relations and Their Fundamental Lemma} 

We are now ready to define a general notion of an operational logical relation. 
Let us consider an operational structure $S$ over a cartesian category $\B$ 
with a strong monad $\mnd$ on $\B$, as in \cref{def:evaluation-structure}. 
Let $\fun\fib\E\B$ be a logical fibration (\cf \cref{def:log-fib}) and $\liftmnd{\mnd}$
be a lifting of $\mnd$ to $\fib$. % $\fun\fib\E\B$
% together with a lifting 
% $\liftmnd{\sem{\gop}}_{\mndfun}$ to $\fun\fib\E\B$, for any generic effect (symbol) $\gop$.

\begin{defi}
\label{def:logical-relation}
    A logical relation is a mapping $\lrel$ from objects of $\syn$ to 
    objects of $\E$ such that $p(\lrel x) = \shallow x$, for any object $x$ of $\syn$,\footnote{
    Actually, it suffices to have 
    $p(\lrel \basetype) = \shallow \basetype$ for basic types $\basetype$ only, as the defining clasues of $\lrel{x}$ assures  the equality 
    $p(\lrel x) = \shallow x$ for all objects $x$.} 
    and the following 
    hold.
    \begin{align*}
        \lrel\unittype &= \top_{\shallow \unittype}
        &
        \lrel(\typeone \times \typetwo) &= \lift{\projl}(\lrel \typeone) \wedge \lift{\projr}(\lrel \typetwo)
        &
        \lrel \comptype{\typeone} &= \lift{\ev}(\liftmnd{\mndfun}(\lrel \typeone))
        \\
        \lrel (\vect{x: \typeone}) &= \liftmnd{\prod} \vect{\lrel \typeone}
        &
         \lrel(\typeone \to \comptype{\typetwo}) &= 
        \uq{}{\pi_1} \lift{\pi_2}(\lrel \typeone) \impl \lift{\app}(\lrel \comptype{\typetwo})
        &
    \end{align*}
\end{defi}

Notice that giving a logical relation essentially amounts to specifying the action of $\lrel$ 
on basic types, since the action of $\lrel$ on complex types is given by \cref{def:logical-relation}. 
The defining clauses of a logical relation 
exploit both the logic of a (logical) fibration and 
the operational semantics of $\lambdacalc$. 
The reader should have recognised in \cref{def:logical-relation} the usual definition of a logical relation, 
properly generalised to rely on the logic of a fibration only. For instance, 
$\lrel \comptype{\typeone}$ intuitively relates computations whose evaluations are related by the lifting 
of the monad. Notice also that the clause of arrow types has a higher logical complexity than other clauses, 
as it involves \emph{two} logical connectives, viz. implication and universal quantification. 
% The following diagram summarises how those relate in the definition of $\lrel(\typeone \to \comptype{\typetwo})$. 

% \[
% \xymatrix{
% \E_{\shallow \typeone} \ar[r]^-{\lift{\pi_1}} 
% & \E_{\shallow( \typeone \to \comptype{\typetwo}) \times \shallow \typeone}
% \ar[d]^-{\uq{}{\pi_1}}
% & \E_{\shallow\comptype{\typeone}} \ar[l]_-{\lift{\app}} 
% \\
% & \E_{\shallow( \typeone \to \comptype{\typetwo})} \ar@/^1pc/[u]^{\lift{\pi_1}}  &
% }
% \]

Operational logical relations come with their so-called \emph{fundamental lemma}, which states 
that (open) terms map (via substitution) related values to related terms. In our abstract framework, 
the fundamental lemma states that to any term $\termone$ we can associate a suitable arrow $\lrel \termone$ in $\E$ lying above $\shallow \termone$. 
% Formally, that means that 
% $\lrel$ extends to a diagram from $\syn$ to $\E$ factorising through $p$ and $\shallow$. 
To prove our general version of the fundamental theorem, we have to assume it for the parameters of our calculus 
(constants of basic types and generic effects). 
Accordingly, we say that a logical relation $\lrel$ is ($\lambdacalc$-)\emph{stable}
if: (i) for every constant $\cbase$ of a base type $\basetype$, we have an arrow 
$\liftmnd{\constarrow} : \dter \to \lrel\basetype$ above $\constarrow$; (ii)
we have an arrow $\liftmnd{\gopsem} : \dter \to \liftmnd{\mndfun}(\lrel\typeone)$ above $\gopsem$. 

\begin{thm}[Fundamental Lemma]
\label{fundamental-lemma}
Let $\lrel$ be a stable logical relation. 
The map $\lrel$ extends to a diagram $\lrel: \syn \to \E$ such that 
$p \circ \lrel = \shallow$.
% \[
% \xymatrix{
% & \E \ar[d]^{p} 
% \\
% \syn \ar[ru]^{\lrel} \ar[r]_{\shallow} & \B 
% }
% \]
In particular, for any term $\envone \compimp \termone: \typeone$, there is an arrow 
$\lrel \termone: \lrel \envone \to \lrel \comptype{\typeone}$ in $\E$ above $\shallow\termone$ (similarly, for values). 
\end{thm}

\begin{proof}[Proof]
Given $\envone \compimp \termone: \typeone$, we construct the desired arrow $\lrel \termone$ by induction on $\termone$. 
The case for values 
lifts commutative triangles 
in \cref{figure:coherence-conditions-operational-structure} using the universal property of cartesian liftings of interaction arrows and then 
constructs the desired arrow using the logical structure of $\fib$ and $\lrel \typeone$. 
%relying on the 'logical' definition of $\lrel \typeone$. 
% 
For the case of terms, we just lift the commutative diagrams 
in \cref{figure:coherence-conditions-operational-structure} using the universal property of the cartesian lifting of $\ev$. 
As a paradigmatic example, we show the case of the sequencing construct $\seq{-}{-}$.
First, let us notice that for any type $\typeone$, the evaluation arrow $\ev: \shallow \comptype{\typeone} \to \mndfun(\shallow \typeone)$ 
gives a cartesian arrow $\bar{\ev}: \lift{\ev}(\liftmnd{\mndfun}(\lrel\typeone)) \to \liftmnd{\mndfun}(\lrel\typeone)$, i.e.
$\bar{\ev}: \lrel{\comptype{\typeone}} \to \liftmnd{\mndfun}(\lrel\typeone)$ (since 
$\lrel{\comptype{\typeone}} = \lift{\ev}(\liftmnd{\mndfun}(\lrel\typeone))$). 
Let us now consider the case of 
$\envone \compimp \seq{\termone}{\termtwo}: \comptype{\typetwo}$ as obtained from 
$\envone \compimp \termone: \typeone$ and $\envone, \varone: \typeone \compimp \termtwo: \typetwo$. 
By induction hypothesis, we have arrows 
$\lrel \termone: \lrel \envone \to \lrel \comptype{\typeone}$ and 
$\lrel \termtwo: \lrel \envone \dtimes \lrel \typeone \to \lrel \comptype{\typetwo}$. 
By postcomposing the former with $\bar \ev$, we obtain the arrow 
$\bar{\ev} \circ \lrel \termone: \lrel \envone \to \liftmnd{\mndfun}(\lrel\typeone)$, and thus 
$\lan \id_{\lrel\envone}, \bar\ev \circ \lrel \termone \ran: \lrel\envone\to  \lrel \envone \dtimes \liftmnd{\mndfun}(\lrel\typeone)$. 
In a similar fashion, we have $\bar{\ev} \circ \lrel \termtwo: \lrel \envone \dtimes \lrel \typeone \to \liftmnd{\mndfun}(\lrel \typetwo)$
and thus, using the extension of the monad, $\liftmnd{\bindsymbol}(\bar{\ev} \circ \lrel \termtwo): 
\lrel \envone \dtimes \liftmnd{\mndfun}(\lrel \typeone) \to  \liftmnd{\mndfun}(\lrel \typetwo)$. 
Altogether, we obtain the arrow 
$\liftmnd{\bindsymbol}(\bar{\ev} \circ \lrel \termtwo) \circ \lan \id_{\lrel\envone}, \bar\ev \circ \lrel \termone \ran: 
\lrel \envone \to  \liftmnd{\mndfun}(\lrel \typetwo)$. Using the commutative diagram of sequencing in \cref{figure:coherence-conditions-operational-structure} and the very definition of a fibration, we obtain
\[
\xymatrix{
\shallow \envone \ar[rd]^-{\hspace{0.2cm} \bindsymbol(\ev \circ \shallow \termtwo) \circ \lan \id_{\envone}, \ev \circ \shallow \termone \ran} 
\ar[d]_{\shallow(\seq{\termone}{\termtwo})} & 
\\
\shallow \comptype{\typetwo} \ar[r]_-{\ev} & \mndfun(\shallow \typetwo)
}
\qquad
\xymatrix{
\lrel \envone \ar[rd]^-{\hspace{0.2cm} 
\liftmnd{\bindsymbol}(\bar{\ev} \circ \lrel \termtwo) \circ \lan \id_{\lrel\envone}, \bar\ev \circ \lrel \termone \ran} 
\ar[d]_{\exists! h} & 
\\
\lrel \comptype{\typetwo} \ar[r]_-{\bar{\ev}} & \liftmnd{\mndfun}(\lrel \typetwo)
}
\]
We choose $h$ as $\lrel(\seq{\termone}{\termtwo})$.
\end{proof}

We can now instantiate \cref{fundamental-lemma} with the operational structures and fibrations seen so far
%(as well as suitable monadic liftings such as the Barr lifting \cite{Barr/LMM/1970} or the more general 
%$\top\top$- and codensity lifting \cite{DBLP:journals/lmcs/KatsumataSU18,DBLP:journals/iandc/Katsumata13})
to recover traditional logical relations (and their fundamental lemmas). For instance, 
the operational structure of \cref{ex:measurable-semantics} and the fibration obtained by pulling back $\Rel\Set$ along 
the forgetful functor from $\Meas$ to \Set  (together with 
the lifting of the Giry monad) give 
operational logical relations for stochastic $\lambda$-calculi, 
whereas considering the fibration $\mathsf{Rel}_{\Set^{\mathcal{L}}}$ and the lifting of the
identity monad for information hiding gives logical relations for non-interference. 
 \cref{fundamental-lemma} then gives 
compositionality (i.e. congruence and substitutivity) of the logical relation.
But that is not the end of the story. In fact, our general results go beyond the realm of traditional logical relations.

%% file: dif-rel.tex
\section{The fibration of differential relations}
\label{sect:dif-rel} 

In this section, we describe the construction of the \emph{fibration of differential relations}, 
which can serve as a fibrational foundation of \emph{differential logical relations} 
\cite{DBLP:conf/icalp/LagoGY19,dlrICTCS,dlrTCS,DalLagoG22} (DLRs, for short), a recently introduced form of logical 
relations defining higher-order distances between programs. DLRs are ternary relations 
relating pairs of terms with elements representing distances between them: such distances, however, 
need not be numbers. More precisely, 
with each type $\typeone$ one associates a set $\dsp{\typeone}$ 
of (higher-order) distances between terms of type $\typeone$, and then defines DLRs as relating terms 
of type $\typeone$ with distances in $\dsp{\typeone}$ between them. Elements of $\dsp{\typeone}$  
reflect the interactive complexity of programs, the latter being given by the type $\typeone$. 
Ignoring effects,\footnote{Without much of a surprise, effects are handled using monads on higher-order distances. 
This, however, becomes clearer when employing generalised distance spaces, which will be introduced soon.} 
we can give a simple inductive definition of $\dsp{\typeone}$ as follows, where 
$\dsp{\zeta}_0$ is a fixed set of distances for base types:
\begin{align*}
    \dsp{\zeta} &= \dsp{\zeta}_0 
    \\
    \dsp{\sigma \times \tau} &= \dsp{\sigma} \times \dsp{\tau} 
    \\
    \dsp{\sigma \to \tau} &= \values_{\sigma} \times \dsp{\sigma} \to \dsp{\tau}.
\end{align*}
Notice that the definition of $\dsp{\sigma \to \tau}$ stipulates that the distance between programs of type 
$\typeone \to \typetwo$ is not \emph{just} a number, 
 but a function $d\termone: \values_\typeone \times \dsp{\typeone} \to \dsp{\typetwo}$ 
 that morally maps a value $\valone$ and an error/perturbation $dv$ to an error/perturbation $d\termone(\valone, dv)$. 

As an example, consider the basic type $\mathbf{R}$ of real numbers 
(we employ the same notation as Example~\ref{ex:measurable-semantics}), and
two values $\lambda x.\termone$ and $\lambda x.\termtwo$ of type
$\mathbf{R} \to \mathbf{R}$. As expected, we define $\dsp{\mathbf{R}}$ as $[0,\infty]$, 
so that elements in $\dsp{\mathbf{R} \to \mathbf{R}}$ are functions
$d\termone: \values_{\mathbf{R}} \times [0,\infty] \to [0,\infty]$. 
Now, not all such functions qualify as valid distances between $\lambda x.\termone$ and 
$\lambda x.\termtwo$, and the purpose of a
differential logical relation is precisely to tell us when 
$d\termone$ indeed acts as a valid distance between $\lambda x.\termone$ and $\lambda x.\termtwo$. 
Accordingly, DLRs are defined as type-indexed ternary relations 
$D_\sigma^{\values} \subseteq \values_\sigma \times \dsp{\sigma} \times \values_{\sigma}$, 
$D_\sigma^{\Lambda} \subseteq \Lambda_\sigma \times \dsp{\sigma} \times \Lambda_{\sigma}$ 
following the rationale that related inputs give related outputs. But what does that mean in 
a differential setting? 

On base types, we need to fix the defining conditions of a DLR in advance. For instance, a natural choice 
in our example is to stipulate 
$$
D^{\values}_{\mathbf{R}}(c_r, \varepsilon, c_{r'}) \iff \varepsilon \geq |r - r'|.
$$
The heart of DLRs is given by the inductive cases, and in particular by the one of function types.
\begin{align*}
    D^{\values}_{\sigma_1 \times \sigma_2}(\lan \valone_1, \valone_2\ran, \lan d\valone_1, d\valone_2\ran, \lan w_1, w_2\ran) 
    & \iff\forall i \in \{1,2\}.\ D^{\values}_{\sigma_i}(\valone_i, d\valone_i, w_i)
    \\
    D^{\values}_{\sigma \to \tau}(\lambda x.t, dt, \lambda x.s) &\iff 
    \forall v,dv,w. (D^{\values}_{\sigma}(v, dv, w) 
    \\
    &\qquad \qquad \implies D^{\Lambda}_{\tau}(t[v/x], dt(v,dv), s[w/x]))
    \\
    D^{\Lambda}_{\sigma}(t,dt,s) &\iff  D^{\values}_{\sigma}(\sem{t}, dt, \sem{s}) 
    \end{align*}
We see that
$d\termone$ is a valid distance between the aforementioned values $\lambda x.\termone$ and $\lambda x.\termtwo$ 
if whenever we have 
two values $\valone, \valtwo$ at distance $d\valone \in [0,\infty]$, then 
$\termone[\valone/x]$ and $\termtwo[\valtwo/x]$ are at distance $d\termone(\valone, d\valone)$. 
If, for instance, we take\footnote{For readability, we use traditional $\lambda$-calculus notation 
rather than the fine-grained one, hence writing, e.g., $\lambda x.v$ in place of 
$\return{(\lambda x.\return{x})}$ and $\lambda x.\return{x}$.} $\lambda x.x$ and $\lambda x. c_2 * x$, then we see that the map 
$dt(c_r, \varepsilon) = |2 \varepsilon - r|$ gives a valid distance between 
$\lambda x.x$ and $\lambda x. c_2 * x$. 

Notice that in DLRs distances are highly interactive and input-dependent, this way going quite far 
from traditional logical relations. An example of that is witnessed by the observation that 
programs usually have non-null self-distances. The identity combinator 
$\lambda x.x$, for instance, has self-distance $dI(v,dv) = dv$, meaning that $\lambda x.x$ 
propagates distances between its inputs to its outputs.\footnote{The reader familiar with 
program metric may have noticed some similarities between self-distances and program sensitivity~\cite{Pierce/DistanceMakesTypesGrowStronger/2010}. There is indeed a link, as self-distances 
can be seen as notions of higher-order sensitivity~\cite{DalLagoG22}.} 
Nevertheless, the meta-theory of DLRs has also several similarities with the theory of traditional 
logical relations. The so-called fundamental lemma of DLR~\cite{DBLP:conf/icalp/LagoGY19}, for example,
states that terms always have a valid self-distance and thus somehow resembles the (traditional) fundamental lemma 
of logical relations, the latter stating that any program is logically related with itself. 
Besides, both these fundamental lemmas are key to compositional reasoning about programs, 
and thus play the same role in their respective frameworks.

% As metric logical relations 
% \cite{Pierce/DistanceMakesTypesGrowStronger/2010}, DLRs are ternary relations: but contrary to 
% their metric counterpart, they do not relate pairs of programs with a number (the latter representing an upper bound of 
% the actual distance between programs). Instead, to each type $\typeone$ one associates a space $\dsp{\typeone}$ 
% of (higher-order) distances between expressions of type $\typeone$, and the defines DLRs as relating terms 
% of type $\typeone$ with a distance in $\dsp{\typeone}$ between them. Elements of $\dsp{\typeone}$ usually 
% reflect the interactive complexity of programs, the latter being given by the type $\typeone$. Thus, for instance, 
% the main novelty of DLRs is that a distance between two expressions of type $\typeone \to \typetwo$ is not \emph{just} a number, 
% but it is a function itself, namely one belonging to $\values_\typeone \times \dsp{\typeone} \to \dsp{\typetwo}$. 
% The informal reading of $\varphi: \values_\typeone \times \dsp{\typeone} \to \dsp{\typetwo}$ is that 
% $\varphi$ maps a value $\valone$ and an error/perturbation $dv$ to an error/perturbation $\varphi(\valone, dv)$. 
% A DLR then relates two values $\valone_1, \valone_2: \typeone \to \typetwo$ with a distance 
% $d \valone$ if whenever we have inputs $\valtwo_1, \valtwo_2$ related with $dw$ (meaning that 
% $\valtwo_1$ and $\valtwo_2$ are $dw$-apart), then $\valone_1\valtwo_1$ and $\valone_2\valtwo_2$
% are related to $dv(\valone_1, dv)$.

The just sketched theory of DLRs is rather syntactic and not very well-suited neither for meta-theoretical developments 
nor for extensions to richer languages, \emph{in primis} those involving computational effects. 
For those reasons, \emph{generalised distance spaces} \cite{DBLP:conf/icalp/LagoGY19,DalLagoG22} have been introduced 
as a semantic counterpart of DLRs with categorical operations between them, such as product and exponential, replacing 
the inductive constructions used for type constructs in the definition of a DLR.

A generalised distance space is defined as a relational structure 
$(X, \dsp X, \delta_X)$ with $\delta_X \subseteq X \times \dsp X \times X$ acting as the 
semantic counterpart 
of a DLR. Such spaces form a category with morphisms between $(X, \dsp X, \delta_X)$ and $(Y, \dsp Y, \delta_Y)$ given by 
pairs of functions $(d, df)$, where $f: X \to Y$ and $df: X \times \dsp X \to \dsp Y$, such that 
$\delta_X(x, dx, x') \implies \delta_Y(f(x), df(x,dx), f(x'))$. 
Given this definition of arrows between generalised distance spaces, it is straightforward 
to define products and exponentials between such spaces. 
This way, we indeed see that generalised distance spaces give a semantic counterpart of DLRs. 
Furthermore, generalised distance spaces ease the understanding of effectful extensions of DLRs. 
In fact, once we have fixed a monad  $\mathbb{T}$ on the category of sets and functions, we simply need to 
\emph{lift} $\mathbb{T}$ (in the categorical sense) to the category of generalised distance spaces. 

In \cite{DalLagoG22} this strategy has been put into practice to define DLRs on languages with algebraic effects. 
The resulting notion --- called a \emph{differential extension} --- however, is not quite a monadic lifting (not apparently, at least),
as one needs to relax the latter notion to cover interesting concrete examples. Even if theoretically unsatisfactory, 
the notion of a differential extension seems interesting in itself, as one can extend the well-known construction 
of a relational extension by Barr \cite{Barr/LMM/1970} to differential extensions, and the result thus obtained has a canonical flavour. 

All of that suggests that differential extensions may be instances of a more general and structural construction, and that 
their associated notion of a Barr-like extension (also called coupling-based extension) is indeed canonical in some suitable sense. These problems have been 
left open in \cite{DalLagoG22}.
Here, we show how our general notion of an operational logical relation 
subsumes the one of a DLR (and, consequently, that the fundamental lemma of DLRs is an instance of \cref{fundamental-lemma}). 
Additionally, we show how differential 
extensions are precisely liftings of monads to the fibration of differential relations and how the so-called coupling-based 
differential extension \cite{DalLagoG22} 
is an instance of a general monadic lifting to such a fibration, viz.  
the well-known Barr lifting properly fitted to a differential setting.

% DLRs have been defined both for pure and effectful calculi, and their semantic foundations has been given 
% in terms of generalised distance spaces and differential extensions of monads. The former are relational structures of 
% the form $(X, \dsp X, \delta_X)$ with $\delta_X \subseteq X \times \dsp X \times X$ acting as the semantic counterpart 
% of a DLR; the latter are reminiscent of extensions of monads to the category of generalised distance spaces, although 
% their mathematical status is still unclear. Here, we show how our general notion of an operational logical relation 
% subsumes the one of a DLRs and how the fundamental lemma of DLRs is an instance of \cref{fundamental-lemma}. 
% The byproduct of all that is twofold: on one hand, our general analysis shows how generalised distance spaces and DLRs 
% are instances of a more general construction (viz. the one of a differential fibration) that works for a general class of 
% fibration, meaning that DLRs can be given even \emph{beyond} sets. On the other hand, 
% our fibrational account sheds a new light on differential extensions of monads: in fact, we show that differential 
% extensions are precisely lifting of monads to differential fibrations and that the coupling-based 
% differential extension of \cite{DalLagoG22} (whose canonicity has been left as an open problem) 
% is an instance of a general monadic lifting to differential fibrations: remarkably, the latter is 
% nothing but the extension of the well-known Barr lifting to a differential setting.

\subsection{Going Differential, Fibrationally}

Let $\fun\fib\E\B$ be a fibration where \B has finite products. 
We define the category $\GDSC\fib$ of \emph{differential relations} in $\fib$ as follows: 
\begin{description}
\item[objects]
are triples $\dX = \ple{\dun\dX,\dsp\dX,\dist\dX}$, where 
$\dun\dX$ and $\dsp\dX$ are objects in \B and $\dist\dX$ is an object in the fibre $\fibre\E{\dun\dX\times\dsp\dX\times\dun\dX}$; 
\item[arrows] 
\fun{\darf}{\dX}{\dY} are triples $\darf = \ple{\dun\darf,\Der\darf,\dpr\darf}$, where 
\fun{\dun\darf}{\dun\dX}{\dun\dY} and \fun{\Der\darf}{\dun\dX\times\dsp\dX}{\dun\dY\times\dsp\dY} are arrows in \B such that $\pi_1\circ\Der\darf = \dun\darf\circ\pi_1$, and 
\fun{\dpr\darf}{\dist\dX}{\dist\dY} is an arrow in \E over $\Der\darf\times\dun\darf$; 
\item[composition] 
of arrows \fun{\darf}{\dX}{\dY} and \fun{\darg}{\dY}{\dZ} is defined componentwise, that is, 
$\dun{\darg\circ\darf} = \dun\darg \circ \dun\darf$, 
$\Der(\darg\circ\darf) = \Der\darg\circ\Der\darf$, and 
$\dpr{\darg\circ\darf} = \dpr\darg\circ\dpr\darf$; 
\item[identity] 
on $\dX$ is given by 
$\id_\dX = \ple{\id_{\dun\dX}, \id_{\dun\dX\times\dsp\dX}, \id_{\dist\dX}}$. 
\end{description}
It is immediate to see that $\GDSC\fib$ is a category. 
Moreover, we have three functors from $\GDSC\fib$ to $\B$: 
\begin{itemize}
\item a functor \fun{\GDS\fib}{\GDSC\fib}{\B} mapping $\dX$ to $\dun\dX$ and $\darf$ to $\dun\darf$; 
\item a functor \fun\Der{\GDSC\fib}{\B} mapping $\dX$ to $\dun\dX\times\dsp\dX$ and $\darf$ to $\Der\darf$; 
\item a functor \fun\nabla{\GDSC\fib}{\B} mapping $\dX$ to $\Der\dX\times\dun\dX$ and $\darf$ to $\Der\darf\times\dun\darf$. 
\end{itemize}

\begin{rem}
Given an arrow \fun\darf\dX\dY\ in $\GDSC\fib$, 
its second component \fun{\Der\darf}{\Der\dX}{\Der\dY}, by the universal property of the product $\Der\dY$,  has shape $\iple{\pi_1\Der\darf,\pi_2\Der\darf}$ and, since $\pi_1\Der\darf = \dun\darf\pi_1$, we get $\Der\darf = \iple{\dun\darf\pi_1, \pi_2\Der\darf}$. 
We call the arrow \fun{\der\darf = \pi_2\Der\darf}{\dun\dX\times\dsp\dX}{\dsp\dY} the \emph{derivative} of $\darf$
and note that, if \fun\darg\dY\dZ\ is another arrow in $\GDSC\fib$, we have 
$\der(\darg\circ\darf) = \der\darg \circ \iple{\dun\darf\pi_1,\der\darf} = \der\darg \circ \Der\darf$. 
Then, the arrow $\darf$ can be equivalently specified by replacing $\Der\darf$ with $\der\darf$. 
This alternative characterisation matches exactly the notion of morphism of generalised distance spaces \cite{DalLagoG22}, but we prefer using the other one, as it is lighter to manage. 
\end{rem}

\begin{exa}\label{ex:gds}
For the fibration \fun{\Sub[\Set]}{\SubC\Set}{\Set} of \cref{ex:subsets}, the category $\GDSC{\Sub[\Set]}$ is the category of generalised distance spaces \cite{DalLagoG22}. 
An object in $\GDSC{\Sub[\Set]}$ is essentially a triple \ple{X,V,R} consisting of a set $X$ of points, a set $V$ of distance values, and a ternary relation $R\subseteq X\times V \times X$ specifying at which distance two elements of $X$ are related, that is, 
$\ple{x,v,y} \in R$ means that $x$ and $y$ are related at distance $v$ or up to an error $v$. 
For instance, a metric \fun{d}{X\times X}{[0,\infty]} on $X$ can be seen as a ternary relation $R_d \subseteq X\times [0,\infty]\times X$ defined by 
$\ple{x,v,y}\in R_d$ iff $d(x,y)\le v$. 
An arrow from \ple{X,V,R} to \ple{Y,U,S} in $\GDSC{\Sub[\Set]}$ consists of a function 
\fun{\dun f}{X}{Y}, transforming points, together with a (parametric) function 
\fun{\der f}{X\times V}{U}, transforming distance values, such that, 
for all $x,y\in X$ and $v\in V$, 
$\ple{x,v,y}\in R$  implies $\ple{\dun f (x), \der f(x,v), \dun f(y)}\in S$. 
\end{exa}

Our goal is to show that the functor \fun{\GDS\fib}{\GDSC\fib}{\B} is a fibration, named the \emph{fibration of differential relations} in $\fib$. 
To this end, we first recognise when an arrow in $\GDSC\fib$ is cartesian in $\GDS\fib$. 

\begin{lem}\label{prop:gds-car}
An arrow \fun\darf\dX\dY\ in $\GDSC\fib$ is cartesian in $\GDS\fib$ if 
the square $\dun\darf \pi_1 = \pi_1\Der\darf$ is a pullback in \B and $\dpr\darf$ is cartesian in $\fib$. 
\end{lem}
\begin{proof}
Consider an arrow \fun\darg\dZ\dY in $\GDSC\fib$ such that $\dun\darg = \dun\darf \circ u$ for some arrow \fun{u}{\dun\dZ}{\dun\dX}. 
By definition, $\darf$ is cartesian iff 
there is a unique arrow \fun{\darh}{\dZ}{\dX} such that $\dun\darh = u$ and $\darg = \darf\circ \darh$, that is, iff 
there are unique arrows \fun{\Der\darh}{\Der\dZ}{\Der\dX} in \B and \fun{\dpr\darh}{\dist\dZ}{\dist\dX} in \E such that 
$\dun\darh \pi_1 = \pi_1\Der\darh$ and $\fib(\dpr\darh) = \Der\darh\times\dun\darh$, and 
$\Der\darg = \Der\darf\circ\Der\darh$ and $\dpr\darg = \dpr\darf\circ\dpr\darh$. 
We take $\Der\darh$ as the unique arrow making the following diagram commute, which exists thanks to the pullback property. 
\[\xymatrix@C=5ex@R=3ex{
\Der\dZ \ar@{..>}[dr]_-{\Der\darh} \ar[drr]^-{\Der\darg} \ar[dd]_-{\pi_1} \\
& \Der\dX \ar[dd]_{\pi_1} \ar[r]_-{\Der\darf} & \Der\dY \ar[dd]^-{\pi_1} \\ 
\dun\dZ \ar[dr]_-{u} \ar[drr]^{\dun\darg} \\
& \dun\dX \ar[r]_-{\dun\darf} & \dun\dY 
}\]
Then, $\dpr\darh$ is the unique arrow over $\Der\darh\times u$ such that $\dpr\darg = \dpr\darf\circ\dpr\darh$, which exists because $\dpr\darf$ is cartesian over $\Der\darf\times\dun\darf$ and $\Der\darg\times\dun\darg = (\Der\darf\times \dun\darf)\circ(\Der\darh\times u)$ holds. 
\end{proof}

\begin{prop}\label{cor:gds-fib}
Let \fun\fib\E\B be a fibration where \B has finite products. 
Then, the functor \fun{\GDS\fib}{\GDSC\fib}{\B} is a fibration. 
\end{prop}
\begin{proof}
Define the cartesian lifting along an arrow \fun{u}{I}{J} of an object $\dY$ with $\dun\dY = J$ as the arrow 
\fun{\liftar{u}{\dY}}{\lift{u}{\dY}}{\dY} where 
$\dun{\liftar{u}{\dY}} = u$, 
$\Der\liftar{u}{\dY} = u\times \id_{\dsp\dY}$, and 
$\dpr{\liftar{u}{\dY}} = \liftar{u\times\id_{\dsp\dY}\times u}{\dist\dY}$, and where 
$\dun{\lift{u}{\dY}} = I$, 
$\dsp{\lift{u}{\dY}} = \dsp\dY$, and 
$\dist{\lift{u}{\dY}} = \lift{(u\times\id_{\dsp\dY}\times u)}{\dist\dY}$. 
\end{proof}

\begin{rem}\label{rem:simple}
The fibration $\GDS\fib$ 
can be also obtained from the simple fibration \fun{\sfib\B}{\simp\B}{\B} (\cf \cref{ex:simp-fib}), 
Indeed, we have that  $\GDS\fib = \sfib\B\circ\fib'$ where 
\fun{\fib'}{\GDSC\fib}{\simp\B} is obtained by the pullback 
\(
\vcenter{\xymatrix{
\GDSC\fib \ar[d]_{\fib'} \ar[r] \ar@{}[dr]|{p.b.} & \E \ar[d]^{\fib} \\ 
\simp\B \ar[r]^{\nabla_s} & \B 
}}\),
with the functor \fun{\nabla_s}{\simp\B}{\B} mapping \ple{I,X} to $I\times X \times I$ and \ple{u,f} to \iple{u\pi_1, f\iple{\pi_1,\pi_2}, u\pi_3}. 
Therefore, we have $\fib'(\dX) = \ple{\dun\dX,\dsp\dX}$ and $\fib'(\darf) = \ple{\dun\darf,\der\darf}$. 
This is somewhat similar to the simple coproduct completion of a fibration \cite{Hofstra11}.  
We leave a precise comparison 
between these constructions for future work. 
\end{rem}

Relying on this characterisation of the fibration of differential relations, 
we can recognise under which conditions it is a logical fibration. 

\begin{prop}\label{prop:gds-lf}
Let \fun\fib\E\B be a fibration. 
\begin{enumerate}
\item\label{prop:gds-lf:1} If $\fib$ has finite products, then $\GDS\fib$ has finite products as well. 
\item\label{prop:gds-lf:2} If $\fib$  is a logical fibration and \B is cartesian closed, then $\GDS\fib$ is a logical fibration. 
\end{enumerate}
\end{prop}

In order to prove the above result, we need an equivalent characterisation of logical fibrations, 
when the base category is carteisan closed. 

\begin{prop}\label{prop:lf-equiv}
Let \fun\fib\E\B be a fibration where \B is cartesian closed. 
Then the following are equivalent:
\begin{itemize}
\item $\fib$ is a logical fibration; 
\item $\fib$ is fibred cartesian closed, \E is cartesian closed and $\fib$ preserves finite products and exponentials. 
\end{itemize}
\end{prop}

This follows  immediately from \cite[Proposition 9.2.4]{Jacobs01}. 
We can now prove our result. 

\begin{proof}[Proof of \cref{prop:gds-lf}] 
By \cref{rem:simple}, we know that $\GDS\fib$ is the composite $\sfib\B\circ\fib'$, where $\fib'$ is obtained by pulling back $\fib$ along \fun{\nabla_s}{\simp\B}{\B}, where $\nabla_s$ preserves products. 

For \cref{prop:gds-lf:1}, since $\simp\B$ has finite products and $\fib$ has finite products, $\fib'$ has finite products as well. 
Then, since \B has finite products, $\sfib\B$ has finite products, thus the composition $\sfib\B\circ\fib'$ has finite products too. 

For \cref{prop:gds-lf:2}, since $\fib$ is a logical fibration and $\nabla_s$ preserves finite products, $\fib'$ is a logical fibration as well. 
Since $\B$ is cartesian closed, by \cref{prop:simp-lf} $\sfib\B$ is a logical fibration.
Therefore, by \cref{prop:lf-equiv}, both $\simp\B$ and $\GDSC\fib$ are cartesian closed categories, and both $\sfib\B$ and $\fib'$ preserve products and exponentials, thus the composition $\sfib\B\circ\fib'$ preserves products and exponentials as well. 
Hence, it remains to show that $\fib$ is fibred cartesian closed. 

To this end, let us fix an object $I$ in $\B$ and consider the functor 
\fun{\fib'_I}{\fibre{\GDSC\fib}}{\fibre{\simp\B}{I}} obtained by restricting and corestricting $\fib'$. 
It is easy to see that $\fib'_I$ is a fibration that is fibred cartesian closed, as fibres and reindexing are those of $\fib'$. 
More precisely, the fibre over \ple{I,X} in $\fib'_I$ is essentially the fibre over $I\times X \times I$ in $\fib$ and 
reindexing along  an arrow \ple{\id_I,f} is given by rendixing in $\fib$ along the arrow $\nabla_s\ple{\id_I,f} = \iple{\id_I,f}\times \id_I$. 
Therefore, given a projection \fun{\ple{\id_I,\pi_2}}{\ple{I,X\times Y}}{\ple{I,X}}, reinexing along it is reindexing along 
\fun{\iple{\pi_1,\pi_2,\pi_4}}{I\times X \times Y\times I}{I\times X \times I} in $\fib$, which has a right adjoint because $\fib$ has universal quantifiers. 
Furthermore, this right adjoint satisfies the Beck-Chevalley condition, as $\nabla_s$ preserves pullbacks along projections. 
This shows that $\fib'_I$ is a logical fibration and, since $\fibre{\simp\B}{I}$ is cartesian closed as it is a fibre of $\sfib\B$ which is fibred cartesian closed, 
by \cref{prop:lf-equiv}, we get that $\fibre{\GDSC\fib}{I}$ is cartesian closed as well. 

To conclude, we have to prove that the cartesian closed structure is preserved by reindexing in $\GDS\fib = \sfib\B\circ \fib'$. 
Let \fun{u}{I}{J} be an arrow in $\B$. 
Since the cartesian closed structure in the fibres $\fibre{\GDSC\fib}{I}$ and $\fibre{\GDSC\fib}{J}$ is built from the logical structure of $\fib'_I$ and $\fib'_J$ and the cartesian closed structure of $\fibre{\simp\B}{I}$ and $\fibre{\simp\B}{J}$, to get our result it suffices to give  a 1-arrow 
\oneAr{\ple{F_u,F'_u}}{\fib'_J}{\fib'_I} where $F_u$ preserves the cartesian closed structure and $F'_u$ preserves the logical structure as well as cartesian arrows. 
We let \fun{F_u}{\fibre{\simp\B}{J}}{\fibre{\simp\B}{I}} be the reindexing functor along $u$ in $\sfib\B$, which maps \ple{J,X} to \ple{I,X}. 
Since $\sfib\B$ is fibred cartesian closed, $F_u$ preserves the cartesian closed structure. 
The functor \fun{F'_u}{\fibre{\GDSC\fib}{J}}{\fibre{\GDSC\fib}{I}} is the reindexing functor along $u$ in $\GDS\fib = \sfib\B\circ \fib'$, 
thus, it follows that $\fib'_I \circ F'_u = F_u\circ \fib'_J$ and $F'_u$ preserves cartesian arrow (see \cite{Streicher18} for details). 
More explicitly, we have 
$F'_u \dX = \ple{I, \dsp\dX, \lift{u\times\id_{\dsp\dX}\times u}\dist{\dX}}$, for all $\dX$ in $\fibre{\GDSC\fib}{J}$. 
Therefore, the restriction of $F'_u$ to the fibre of $\fib'_J$ over \ple{J,X} is essentially given by reindexing along $u\times\id_X\times u$ in $\fib$, thus it preserves the cartesian closed structure of the fibres, as $\fib$ is fibred cartesian closed. 
Finally, $F'_u$ preserves universal quantifiers thanks to the Beck-Chevalley condition in $\fib$, since the following diagram is a pullback.

  \begin{minipage}[b]{\linewidth-2cm}
  \begin{align*}
  \xymatrix@C=12ex{
  I\times X \times Y \times I 
    \ar[r]^-{u\times \id_{X\times Y} \times u}
    \ar[d]_-{\iple{\pi_1,\pi_2,\pi_4}} 
& J\times X \times Y \times J 
    \ar[d]^-{\iple{\pi_1,\pi_2,\pi_4}} 
\\ 
  I\times X \times I 
    \ar[r]^-{u\times \id_X\times u}
& J \times X \times J 
  } 
  \end{align*}
  \end{minipage}
\end{proof}

We report below the explicit definition of the logical structure on $\GDS\fib$. 
 In the following, for $I,J$ objects in the base category of $\fib$, we will write $\expobj{I}{J}$ for their exponential and 
$\fun{\evnt^I_J}{\expobj{I}{J}\times I}{J}$ for the associated evaluation arrow.  
Let $\dX,\dY,\dZ$ be objects in $\GDSC\fib$ with $\dun\dX = \dun\dY = I$ and $\dun\dZ = I\times J$, then we have:
\begin{align*}
&\what\top_I \eqdef \ple{I, \terobj, \top_{I\times\terobj\times I}} \\ 
&\dX\what\land_I\dY \eqdef \ple{I, \dsp\dX\times\dsp\dY, \lift{\iple{\pi_1,\pi_2,\pi_4}}{\dist\dX}\land\lift{\iple{\pi_1,\pi_3,\pi_4}}{\dist\dY}} 
  \\ 
&\dX \what\impl_I \dY \eqdef \ple{I, \expobj{\dsp\dX}{\dsp\dY}, 
    \uqn{\iple{\pi_1,\pi_2,\pi_4}} ( \lift{\iple{\pi_1,\pi_3,\pi_4}}{\dist\dX} \impl\lift{\iple{\pi_1,\evnt^{\dsp\dX}_{\dsp\dY}\iple{\pi_2,\pi_3}, \pi_4}}{\dist\dY} )}  \\ 
&\what{\uq I J}\dZ \eqdef \ple{I, \expobj{J}{\dsp\dZ}, 
    \uqn{\iple{\pi_1,\pi_3,\pi_5}} ( \lift{\iple{\pi_1,\pi_2,\evnt^J_{\dsp\dZ}\iple{\pi_3,\pi_2},\pi_4,\pi_5}}{\dist\dZ} )}.
\end{align*}

\begin{exa}
Let us consider the fibration \fun{\GDS{\Sub[\Set]}}{\GDSC{\Sub[\Set]}}{\Set} of set-theoretic differential relations 
and instantiate the above constructions in this case. 
The terminal object over a set $X$ 
is $\what\top_X = \ple{X, 1, X\times 1 \times X}$; that is, $\what\top_X$ has just one distance value and all elements of $X$ are related at that value.
Consider now objects $\dX = \ple{X,V,R}$ and $\dY = \ple{X,U,S}$ in $\GDSC{\Sub[\Set]}$ over the same set $X$.  
Then, their product and exponential over $X$ are given by  
$\dX\what\land\dY = \ple{X, V\times U, R\sqcap S}$ and 
$\dX\what\impl\dY = \ple{X,\expobj{V}{U}, R\to S}$, where 
\begin{itemize}
\item $\ple{x,\ple{v,u},y} \in R\sqcap S$ iff $\ple{x,v,y} \in R$ and $\ple{x,u,y}\in S$ and 
\item $\ple{x,f,y} \in R\to S$ iff, for all $v\in V$, $\ple{x,v,y}\in R$ implies $\ple{x, f(v), y} \in S$. 
\end{itemize}
That is, two elements $x$ and $y$ are related in $R\sqcap S$ at a distance \ple{v,u} if and only if 
$x$ and $y$ are at distance $v$ in $R$ and $u$ in $S$.
By contrast, $x$ and $y$ are related in $R\to S$ at a distance $f$, which transforms distances in $V$ into distances in $U$, 
iff, whenever $x$ and $y$ are at distance $v$ in $R$, they are at distance $f(v)$ in $S$. 
 
Finally, if $\dZ = \ple{X\times Y, V,R}$ is an object in $\GDSC{\Sub[\Set]]}$ over the set $X\times Y$, 
then the universal quantifier on $Y$ is given by 
$\what{\uq X Y}\dZ = \ple{X, \expobj{Y}{V}, Y\to R}$, where 
$\ple{x,f,x'}\in Y\to R$, iff for all $y,y' \in Y$, $\ple{\ple{x,y},f(y),\ple{x',y'}} \in R$. 
That is, elements $x$ and $x'$ are related by $Y\to R$ at a distance $f$, which returns for every element of $Y$ a distance in $V$, if and only if $R$ relates \ple{x,y} and \ple{x',y'} at distance $f(y)$, for each $y,y' \in Y$. 
\end{exa}

We now extend  the construction of the fibration of differential relations  to 1- and 2-arrows. 
To do so, we work with bifibrations whose bases have finite products. 
 This restriction to bifibration is not specific to our differential setting.  
Indeed, besides the differential flavour, our construction turns fibrations of (unary) predicates into fibrations of (ternary) relations, viewed as predicates over a product. 
Then, in order to accordingly transform 1-arrows, we need these assumptions to avoid requiring 1-arrows to preserve products in the base, which is a quite strong condition, especially for monads. 

Let \fun\fib\E\B and \fun\fibq\F\CC be bifibrations where \B and \CC have finite products and 
let \fun{K}{\B}{\CC} be a functor. 
For every object $\dX$ in $\GDSC\fib$, we can consider the arrow 
\fun{\projar{K}{\dX} = \iple{K\pi_1,K\iple{\pi_1,\pi_2},K\pi_3}}{K(\dun\dX\times\dsp\dX\times\dun\dX)}{K\dun\dX\times K(\dun\dX\times\dsp\dX)\times K\dun\dX} 
in \CC splitting $K(\nabla\dX)$ into three components using the images along $K$ of certain projections. 
Using cocartesian liftings of such arrows, we can split arrows in \F whose image along $\fibq$ connects two domains of such splitting arrows. 
 More precisely, 
given functors \fun{K,K'}{\B}{\CC} and objects $\dX,\dY$ of $\GDSC\fib$, 
for every arrow $\alpha$ in \F and every arrow  $\beta$ in $\CC$ such that the diagram on the right below commutes, 
we have a unique arrow $\pick[\beta]{\alpha}$ in \F above $\beta$ making the diagram on the left below commute as well.  
\[
\vcenter{\xymatrix@C=9ex{
  A \ar[d]_-{\alpha} 
    \ar[r]^-{\coliftar{\projar{K}{\dX}}{A}} 
& \colift{\projar{K}{\dX}}{A} 
    \ar@{..>}[d]^-{\pick[\beta]{\alpha}} 
\\ 
  B \ar[r]_-{\coliftar{\projar{K'}{\dY}}{B}} 
& \colift{\projar{K'}{\dY}}{B}
}}  \quad 
\vcenter{\xymatrix{
  K(\nabla\dX)
    \ar[d]_{\fibq(\alpha)}  
    \ar[r]^-{\projar{K}{\dX}} 
& K\dun\dX\times K(\Der\dX) \times K\dun\dX 
    \ar[d]^-{\beta} 
\\
  K'(\nabla\dY)
    \ar[r]_-{\projar{K'}{\dY}} 
& K'\dun\dY \times K'(\Der\dY) \times K'\dun\dY 
}}
\]
 Therefore, 
the diagram on the left lies above the one on the right through the bifibration $\fibq$ and 
the arrow $\pick[\beta]{\alpha}$ is well-defined and uniquely determined thanks to the universal property of the cocartesian arrow  $\coliftar{\projar{K}{\dX}}{A}$. 
In the following, to lighten the notation, we omit the superscript of $\pick[\beta]{\alpha}$ when it is clear from the context.  

The operator $\pick{\blank}$ has the following key property. 

\begin{prop}\label{prop:pick}
$\pick[\id]{\id} = \id$ and $\pick[\beta\circ\beta']{\alpha\circ\alpha'} = \pick[\beta]{\alpha'}\circ\pick[\beta']{\alpha'}$. 
\end{prop}
\begin{proof}
This is immediate by the universal property of cocartesian arrows. 
Indeed, $\id$ and $\pick[\beta]{\alpha}\circ\pick[\beta']{\alpha'}$ make the diagram on the left above commute when $\id$ and $\alpha\circ\alpha'$ replace the leftmost arrow in the same diagram, respectively. 
\end{proof}

Let us now consider a 1-arrow 
\oneAr{F}{\fib}{\fibq} in \biFib. 
We define a functor \fun{\what F}{\GDSC\fib}{\GDSC\fibq} as follows: 
for $\dX$ an object in $\GDSC\fib$, we set  
\[
\dun{\what F \dX}  \eqdef \fb{F}\dun\dX \qquad 
\dsp{\what F \dX}  \eqdef \fb{F}(\dun\dX\times\dsp\dX) = \fb{F}(\Der\dX) \qquad 
\dist{\what F \dX} \eqdef \colift{\projar{\fb{F}}{\dX}}{\ft{F}\dist\dX},
\]
For every arrow \fun{\darf}{\dX}{\dY} in $\GDSC\fib$, we set 
\[
\dun{\what F \darf} \eqdef \fb{F}\dun\darf \qquad 
\Der(\what F \darf) \eqdef \fb{F}\dun\darf \times \fb{F}(\Der\darf)  \qquad 
\dpr{\what F \darf} \eqdef \pick{\ft{F}\dpr\darf}  
\]
where $\pick{\ft{F}\dpr\darf}$ is obtained as above using the following diagrams: 
\[
\vcenter{\xymatrix@C=9ex{
  \ft{F}\dist\dX 
    \ar[d]_-{\ft{F}\dpr\darf} 
    \ar[r]^-{\coliftar{\projar{\fb{F}}{\dX}}{\ft{F} \dist\dX}} 
& \dist{\what F \dX} 
    \ar@{..>}[d]^-{\pick{\ft{F}\dpr\darf}} 
\\ 
  \ft{F}\dist\dY 
    \ar[r]_-{\coliftar{\projar{\fb{F}}{\dY}}{\ft{F} \dist\dY}} 
& \dist{\what F \dY} 
}}  \quad 
\vcenter{\xymatrix{
  \fb{F}(\nabla\dX) 
    \ar[d]_{\fb{F}(\nabla\darf)}  
    \ar[r]^-{\projar{\fb{F}}{\dX}} 
& \fb{F}\dun\dX\times \fb{F}(\dun\dX\times\dsp\dX)\times \fb{F}\dun\dX 
    \ar[d]^{\fb{F}\dun\darf\times \fb{F}(\Der\darf) \times \fb{F}\dun\darf} 
\\
  \fb{F}(\nabla\dY) 
    \ar[r]_-{\projar{\fb{F}}{\dY}} 
& F\dun\dY\times F(\dun\dY\times\dsp\dY)\times F\dun\dY 
}}
\]
Notice that cocartesian liftings are essential to appropriately define $\what F$ on the relational part of $\dX$ and $\darf$, as we do not assume any compatibility with products for $\fb{F}$. 

\begin{prop}\label{prop:gds-1ar} 
 Let \oneAr{F}{\fib}{\fibq} be a 1-arrow in \biFib.
Then,  \oneAr{\ple{\fb{F},\what F}}{\GDS\fib}{\GDS\fibq} is a 1-arrow in \Fib. 
\end{prop} 
\begin{proof}
The equation $\fb{F}\circ\GDS\fib = \GDS\fibq \circ \what F$ trivially holds by definition of $\what F$. 
Hence, to conclude the proof, we only have to check that $\what F$ is indeed a functor  from $\GDSC\fib$ to $\GDSC\fibq$.  
Note that $\what F (\id_\dX) = \ple{\fb{F}\id_{\dun\dX}, \fb{F}\id_{\dun\dX}\times \fb{F}\id_{\Der\dX}, \pick{\ft{F}\id_{\dist\dX}}}$, which by functoriality of $\fb{F}$ and $\ft{F}$ and \cref{prop:pick},
is the identity on $\what F \dX$. 
Preservation of composition follows from a similar argument. 
\end{proof}

Similarly, let us consider a 2-arrow \twoAr{\phi}{F}{G} between 1-arrows \oneAr{F,G}{\fib}{\fibq}. 
We define a natural transformation \nt{\what \phi}{\what F}{\what G} as follows: 
for every object $\dX$ in $\GDS\fib$,  we set 
\[
\dun{\what \phi_\dX} \eqdef \fb{\phi}_{\dun\dX} \qquad 
\Der(\what \phi_\dX) \eqdef \fb{\phi}_{\dun\dX}\times \fb{\phi}_{\dun\dX\times\dsp\dX} \qquad 
\dpr{\what \phi_\dX} \eqdef \pick{\ft{\phi}_{\dist\dX}} 
\]
where $\pick{\ft{\phi}_{\dist\dX}}$ is obtained as above using the following diagrams 
\[
\vcenter{\xymatrix@C=9ex{
  \ft{F}\dist\dX 
    \ar[d]_-{\ft{\phi}_{\dist\dX}} 
    \ar[r]^-{\coliftar{\projar{\fb{F}}{\dX}}{\ft{F}\dist\dX}} 
& \dist{\what F \dX} 
    \ar@{..>}[d]^-{\pick{\ft{\phi}_{\dist\dX}}} 
\\ 
  \ft{G}\dist\dX 
    \ar[r]_-{\coliftar{\projar{\fb{G}}{\dX}}{\ft{G}\dist\dX}}  
& \dist{\what G \dX} 
}} \quad 
\vcenter{\xymatrix{
  \fb{F}(\nabla\dX) 
    \ar[d]_-{\fb{\phi}_{\nabla\dX}} 
    \ar[r]^-{\projar{\fb{F}}{\dX}} 
& \fb{F}\dun\dX \times \fb{F}(\dun\dX\times\dsp\dX) \times \fb{F}\dun\dX 
    \ar[d]^-{\fb{\phi}_{\dun\dX}\times\fb{\phi}_{\dun\dX\times\dsp\dX}\times\fb{\phi}_{\dun\dX}} 
\\ 
  \fb{G}(\nabla\dX) 
    \ar[r]_-{\projar{\fb{G}}{\dX}} 
& \fb{G}\dun\dX \times \fb{G}(\dun\dX\times\dsp\dX) \times \fb{G}\dun\dX  
}}
\] 

\begin{prop}\label{prop:gds-2ar}
 Let \twoAr{\phi}{F}{G} be a 2-arrow in \biFib. 
Then,  \twoAr{\ple{\fb{\phi},\what\phi}}{\ple{\fb{F},\what F}}{\ple{\fb{G},\what G}} is a 2-arrow in \Fib. 
\end{prop}
\begin{proof}
The fact that $\fb{\phi}\GDS\fib = \GDS\fibq \what\phi$ trivially holds by construction of $\what\phi$. 
Hence, we have only to check that $\nt{\what\phi}{\what F}{\what G}$ is indeed a natural transformation. 
To this end, let \fun\darf\dX\dY be an arrow in $\GDSC\fib$. 
We have to verify that $\what\phi_\dY \circ \what F\darf = \what G\darf \circ \what\phi_\dX$ holds, that is, 
\begin{itemize} 
\item $\fb{\phi}_{\dun\dY}\circ \fb{F}\dun\darf = \fb{G}\dun\darf \circ \fb{\phi}_{\dun\dX}$ and 
\item $(\fb{\phi}_{\dun\dY}\times\fb{\phi}_{\Der\dY})\circ(\fb{F}\dun\darf\times \fb{F}(\Der\darf)) = (\fb{G}\dun\darf\times\fb{G}(\Der\darf))\circ(\fb\phi_{\dun\dX}\times\fb{\phi}_{\Der\dX})$ and 
\item $\pick{\ft{\phi}_{\dist\dY}}\circ \pick{\ft{F}\dpr\darf} = \pick{\ft{G}\dpr\darf} \circ \pick{\ft{\phi}_{\dist\dX}}$. 
\end{itemize} 
The first two equations hold by naturality of $\fb{\phi}$. 
The third one holds because, by \cref{prop:pick}, both sides of the equation are equal to $\pick[\beta]{\alpha}$  defined by the diagrams below. 
\begin{mathpar}
\vcenter{\xymatrix{
  \fb{F}(\nabla\dX) 
    \ar[d]_-{\fb{\phi}_{\nabla\dX}}
    \ar[r]^-{\fb{F}(\nabla\darf)}
    \ar[rd]^{\fibq(\alpha)}
& \fb{F}(\nabla\dY) 
    \ar[d]^-{\fb{\phi}_{\nabla\dY}} 
\\
  \fb{G}(\nabla\dX) 
    \ar[r]_-{\fb{G}(\nabla\darf)}
& \fb{G}(\nabla\dY) 
}} \and 
\vcenter{\xymatrix@C=15ex{
  \fb{F}\dun\dX\times\fb{F}(\Der\dX)\times\fb{F}\dun\dX 
    \ar[d]_-{\fb\phi_{\dun\dX}\times\fb\phi_{\Der\dX}\times\fb\phi_{\dun\dX}}
    \ar[r]^-{\fb{F}\dun\darf\times\fb{F}(\Der\darf)\times\fb{F}\dun\darf}
    \ar[rd]^-{\beta} 
& \fb{F}\dun\dY\times\fb{F}(\Der\dY)\times\fb{F}\dun\dY 
    \ar[d]^-{\fb\phi_{\dun\dY}\times\fb\phi_{\Der\dY}\times\fb\phi_{\dun\dY}}
\\
  \fb{G}\dun\dX\times\fb{G}(\Der\dX)\times\fb{G}\dun\dX 
    \ar[r]_-{\fb{G}\dun\darf\times\fb{G}(\Der\darf)\times\fb{G}\dun\darf}
& \fb{G}\dun\dY\times\fb{G}(\Der\dY)\times\fb{G}\dun\dY 
}} \and 
\vcenter{\xymatrix{
  \fb{F}(\nabla\dX) 
    \ar[d]_-{\fibq(\alpha)}
    \ar[r]^-{\projar{\fb{F}}{\dX}}
& \fb{F}\dun\dX\times\fb{F}(\Der\dX)\times\fb{F}\dun\dX 
    \ar[d]^-{\beta}
\\
  \fb{G}(\nabla\dY) 
    \ar[r]_-{\projar{\fb{G}}{\dY}} 
& \fb{G}\dun\dY\times\fb{G}(\nabla\dY)\times \fb{G}\dun\dY 
}} \and 
\vcenter{\xymatrix{
  \ft{F}\dist\dX
    \ar[d]_-{\ft\phi_{\dist\dX}}
    \ar[r]^-{\ft{F}\dpr\darf}
    \ar[rd]^-{\alpha}
& \ft{F}\dist\dY
    \ar[d]^-{\ft\phi_{\dist\dY}}
\\
  \ft{G}\dist\dX
    \ar[r]_-{\ft{G}\dpr\darf}
& \ft{G}\dist\dY
}} \and 
\vcenter{\xymatrix{
  \ft{F}\dist\dX 
    \ar[d]_-{\alpha}
    \ar[r]^-{\coliftar{\projar{\fb{F}}{\dX}}{\ft{F}\dist\dX}}
& \dist{\what F \dX}
    \ar@{..>}[d]^-{\pick[\beta]{\alpha}}
\\
  \ft{G}\dist\dY
    \ar[r]_-{\coliftar{\projar{\fb{G}}{\dY}}{\ft{G}\dist\dY}}
& \dist{\what G \dY}
}}
\vspace*{(-\baselineskip*2)+5pt}
\end{mathpar}
\end{proof}

The construction of $\what\phi$ respects vertical composition and identities of 2-arrows, as the following proposition shows

\begin{prop}\label{prop:gds-2ar-comp}
Let $\fib$ and $\fibq$ be 
bifibrations with finite products in the base. 
\begin{enumerate}
\item\label{prop:gds-2ar-comp:1}
For every 1-arrow \oneAr{F}{\fib}{\fibq}, 
we have that $\nt{\what\id}{\what F}{\what F}$ is the identity natural transformation. 
\item\label{prop:gds-2ar-comp:2}
For 1-arrows \oneAr{F,G,H}{\fib}{\fibq} and 2-arrows \twoAr{\phi}{F}{G} and \twoAr{\psi}{G}{H}, 
we have $\what{\psi\phi} = \what\psi\what\phi$. 
\end{enumerate}
\end{prop}
\begin{proof}
 For both items,  
the equality on the first two components is straightforward in both cases, and equality on the third one follows by \cref{prop:pick}. 
\end{proof}

This shows that, given bifibrations $\fib$ and $\fibq$ with finite products in the base, 
we have a functor 
\fun{\DRF_{\fib,\fibq}}{\biFib(\fib,\fibq)}{\Fib(\GDS\fib,\GDS\fibq)}
given by 
\[
\DRF_{\fib,\fibq}(F)    \eqdef \ple{\fb{F},\what F}
\qquad 
\DRF_{\fib,\fibq}(\phi) \eqdef \ple{\fb\phi,\what\phi}
\] 
We are getting closer and closer to the definition of a 2-functor 
$\DRF$, with $\DRF(\fib) \eqdef  \GDS\fib$,  from (certain) bifibrations to fibrations. 
However, this is not the case, as $\DRF$ does not preserve composition and identity of 1-arrows. 
Indeed, given a bifibration \fun\fib\E\B where \B has finite products and the identity 1-arrow \oneAr\Id\fib\fib on it, 
we have that 
for every object $\dX$ in $\GDSC\fib$, 
$\what \Id(\dX) = \ple{\dun\dX, \dun\dX\times\dsp\dX, \colift{\projar{\fb\Id}{\dX}}{\dist\dX}}$,
which is not isomorphic to $\dX$, in general. 

We can recover a weak form of functoriality by restricting to a 2-subcategory of \biFib, that is, we consider only 1-arrows which are cocartesian, \ie, that preserve cocartesian arrows. 
% We say that a 1-arrow \oneAr{\ple{F,G}}{\fib}{\fibq} in \biFib is \emph{cocartesian} if 
% the functor $G$ preserves cocartesian arrows. 
% This implies that $G$ commutes with cocartesian liftings up to isomorphism. 
Denote by \biFibc the 2-full 2-subcategory of \biFib where objects are bifibrations with finite products in the base and 1-arrows preserve cocartesian arrows. 
Then, restricting $\DRF$ to \biFibc, we get a \emph{lax functor}. 
This means that horizontal composition and identities of 1-arrows are preserved only up to a mediating 2-arrow which needs not be an isomorphism. 

\begin{thm}\label{thm:gds-lax-fun}
\fun{\DRF}{\biFibc}{\Fib} is a lax functor. 
\end{thm}
\begin{proof}
Let $\fib$ be a bifibration in \biFibc. 
We consider a natural transformation $\nt{\xi^\fib}{\ft{\Id_\fib}}{\what{\Id_\fib}}$ defined as follows: 
for every object $\dX$ in $\GDSC\fib$ we set:
\[
\dun{\xi^\fib_\dX} \eqdef \id_{\dun\dX} \qquad 
\Der(\xi^\fib_\dX) \eqdef \iple{\pi_1,\id_{\Der\dX}} \qquad 
\dpr{\xi^\fib_\dX} \eqdef \coliftar{\projar{\Id}{\dX}}{\dist\dX}
\]
Naturality of $\xi^\fib$ is straightforward. 
Moreover, we have $\fib(\xi^\fib_\dX) = \id_{\dun\dX}$, hence 
\twoAr{\ple{\id,\xi^\fib}}{\Id_\fib}{\DRF(\Id_\fib)} is a well-defined 2-arrow in \biFibc. 
Now, for 1-arrows 
\oneAr{F}{\fib}{\fibq} and \oneAr{G}{\fibq}{\fibr} in \biFibc, 
we define a natural transformation 
\nt{\omega^{F,G}}{\what G \circ \what F}{\what {G\circ F}} as follows: 
for every object $\dX$ in $\GDSC\fib$  we set 
\[
\dun{\omega^{F,G}_\dX} \eqdef \id_{\fb{G}\fb{F}\dun\dX} \qquad 
\Der(\omega^{F,G}_\dX) \eqdef \id_{\fb{G}\fb{F}\dun\dX} \times \fb{G}\pi_2 \qquad 
\]
and the third component is given by the following diagrams
\begin{mathpar}
\vcenter{\xymatrix{
  \fb{G}(\fb{F}\dun\dX \times \fb{F}(\dun\dX\times\dsp\dX) \times \fb{F}\dun\dX) 
    \ar[r]^-{\projar{\fb{G}}{\what F \dX}} 
& \fb{G}\fb{F}\dun\dX \times \fb{G}(\fb{F}\dun\dX \times \fb{F}(\dun\dX\times\dsp\dX)) \times \fb{G}\fb{F}\dun\dX 
    \ar[d]^-{\id_{\fb{G}\fb{F}\dun\dX} \times \fb{G}\pi_2 \times \id_{\fb{G}\fb{F}\dun\dX}}
\\
  \fb{G}\fb{F}(\dun\dX\times\dsp\dX\times\dun\dX) 
    \ar[r]_-{\projar{\fb{G}\fb{F}}{\dX}} 
    \ar[u]^-{\fb{G}\projar{\fb{F}}{\dX}} 
& \fb{G}\fb{F}\dun\dX \times \fb{G}\fb{F}(\dun\dX\times\dsp\dX) \times \fb{G}\fb{F}\dun\dX 
}}  \and 
\vcenter{\xymatrix@C=9ex{
  \ft{G}(\dist{\what F\dX})  
    \ar[r]^-{\coliftar{\projar{\fb{G}}{\what F\dX}}{ \dist{\what F \dX}  }} 
& \dist{\what G \what F \dX} 
    \ar@{..>}[d]^-{\dpr{\omega^{F,G}_\dX}} 
\\
  \ft{G}\ft{F}\dist\dX 
    \ar[r]_-{\coliftar{\projar{\fb{G}\fb{F}}{\dX}}{\dist\dX}} 
    \ar[u]^-{\ft{G}\coliftar{\projar{\fb{F}}{\dX}}{\dist\dX}}
& \dist{\what{GF}\dX} 
}} 
\end{mathpar}
where $\dpr{\omega^{F,G}_\dX}$ exists and is uniquely determined since 
the composite $\coliftar{\projar{\fb{G}}{\what F \dX}}{ \dist{\what F \dX} } \circ \ft{G}( \coliftar{\projar{\fb{F}}{\dX}}{\dist\dX} )$ is cocartesian as $\ft{G}$ preserves cocartesian arrows because $G$ is a 1-arrow in \biFibc. 
Naturality can be again easily checked and, since $\fibr(\omega^{F,G}_\dX) = \id_{\dun\dX}$, 
\twoAr{\ple{\id,\omega^{F,G}}}{\DRF(G)\circ\DRF(F)}{\DRF(G\circ F)} is a 2-arrow in \Fib. 
It is also easy to see this 2-arrow is natural in $F$ and $G$. 

To conclude the proof we have to check that the following equations hold
where \oneAr{F}{\fib}{\fibq}, \oneAr{G}{\fibq}{\fibr} and \oneAr{H}{\fibr}{\fibr'} are 1-arrows in \biFibc: 
\begin{align*} 
& \omega^{\Id_\fib,F} \cdot \what F \xi^\fib = \id_{\what F}  \\ 
& \omega^{F,\Id_\fibq} \cdot \xi^\fibq \what F = \id_{\what F}  \\ 
& \omega^{GF,H} \cdot \what H \omega^{F,G} = \omega^{F,HG} \cdot \omega^{G,H} \what F 
\end{align*} 
We only check the first equality, the others can be proved following a similar strategy. 
The equality on the first component is trivial; hence we only check those on the second and the third ones. 
We have 
\begin{align*}
& \Der(\omega^{\Id_\fib,F}_\dX \circ \what F \xi^\fib_\dX) 
    = ( \id_{\fb{F}\dun\dX}\times \fb{F}\pi_2) \circ (\fb{F}\id_{\dun\dX}\times \fb{F}\iple{\pi_1,\id_{\Der\dX}}) \\ 
&\qquad 
    = \id_{\fb{F}\dun\dX} \times \fb{F}(\pi_2\iple{\pi_1,\id_{\Der\dX}})
    = \id_{\fb{F}\dun\dX} \times \id_{\fb{F}(\Der\dX)}
    = \Der(\id_{\what F \dX} 
\\ 
& \dpr{\omega^{\Id_\fib,F}_\dX \circ \what F \xi^\fib_\dX} 
    = \dpr{\omega^{\Id_\fib,F}_\dX} \circ \pick{\ft{F}\coliftar{\projar{\Id}{\dX}}{\dist\dX}}    
    = \id_{\dist{\what F \dX}} 
\end{align*}
The last equality holds using the following commutative diagram, 
where $\coliftar{\projar{\fb{F}}{\dX}}{\ft{F}\dist\dX}$ is cocartesian. 
\begin{minipage}[b]{\linewidth-2cm}
  \begin{align*}
\xymatrix@C=25ex{
  \ft{F}\dist\dX 
    \ar@/_30pt/[dd] 
    \ar[d]^-{\ft{F}\coliftar{\projar{\Id}{\dX}}{\dist\dX}} 
    \ar[r]^-{\coliftar{\projar{\fb{F}}{\dX}}{\ft{F}\dist\dX}} 
& \dist{\what F \dX} 
    \ar@/^30pt/[dd]^-{\id_{\dist{\what F\dX}}} 
    \ar[d]_-{\pick{\ft{F}\coliftar{\projar{\Id}{\dX}}{\dist\dX}}} 
\\
  \ft{F}\dist{\what\Id\dX} 
    \ar[r]_-{\coliftar{\projar{\fb{F}}{\what\Id\dX}}{\ft{F}\dist{\what\Id\dX}} } 
& \dist{\what F\what \id \dX}
    \ar[d]_-{\dpr{\omega^{\Id^\fib,F}_\dX}} 
\\ 
  \ft{F}\dist\dX 
    \ar[u]_-{\ft{F}\coliftar{\projar{\Id}{\dX}}{\dist\dX}} 
    \ar[r]_-{\coliftar{\projar{\fb{F}}{\dX}}{\ft{F}\dist\dX}} 
& \dist{\what F \dX} 
}
\end{align*}
\end{minipage}
\end{proof}

This is a quite weak form of functoriality, 
nevertheless, this is enough to get important properties of  the construction of the fibration of differential relations.

\subsection{Lifting monads: the differential extension}
\label{sect:dif-mnd}

An important problem when dealing with effectful languages is the lifting of monads from the category where the semantics of the language is expressed to the category used to reason about programs, e.g., the domain of a fibration or a category of relations. 
The most famous lifting is perhaps the so-called \emph{Barr extension} 
\cite{Barr/LMM/1970} of a \Set-monad to the category of (endo)relations, which is fibred over \Set 
(other notions of lifting include $\top\top$- and codensity lifting \cite{DBLP:journals/lmcs/KatsumataSU18,DBLP:journals/iandc/Katsumata13}).
In the differential setting, the notion of a \emph{differential extension} has been recently proposed 
\cite{DalLagoG22} as a way to lift monads to generalised distance spaces. 
To what extent such a construction 
is canonical and whether it defines an actual monadic lifting have, however, been left as open questions. 
Here, we answer both questions in the affirmative.

\begin{defi}\label{def:gds-dif-ext}
Let $\fun\fib\E\B$ be a fibration where \B has finite products, and 
let $\mnd = \ple{{T},\mu,\eta}$ be a monad on \B. 
A \emph{differential extension} of $\mnd$ along $\fib$ is a lifting 
$\liftmnd\mnd = \ple{\liftmnd T, \liftmnd \mu, \liftmnd \eta}$ of $\mnd$ along the fibration \fun{\GDS\fib}{\GDSC\fib}{\B}. 
\end{defi}
A differential extension of $\mnd$ along the fibration $\fib$ is thus a monad on the category $\GDSC\fib$ of differential relations in $\fib$ 
which is above $\mnd$ with respect to the fibration $\GDS\fib$. 
We now describe a way to build differential extensions of monads starting from a usual monadic extension. 
This is an immediate consequence of \cref{thm:gds-lax-fun}, since 
$\fun\DRF\biFibc\Fib$ is a lax functor and thus it preserves monads. 
Given a monad $\mnd = \ple{T,\mu\eta}$ on a bifibration \fun\fib\E\B in \biFibc, recall that $\fb\mnd = \ple{\fb{T},\fb\mu,\fb\eta}$ is a monad on \B and $\ft\mnd = \ple{\ft{T},\ft\mu,\ft\eta}$ is a lifting of $\fb\mnd$ along $\fib$ where \fun{\ft{T}}{\E}{\E} preserves cocartesian arrows. 
Applying $\DRF$, we get a monad $\DRF(\mnd)$ on the fibration $\GDS\fib$ such that $\fb{\DRF(\mnd)} = \fb\mnd$, hence $\ft{\DRF(\mnd)}$ is a differential extension of $\fb\mnd$ along $\fib$. 
This follows from general properties of lax functors \cite{Benabou67,Street72lax}. 

\begin{cor}\label{cor:gds-mnd}
Let $\mnd = \ple{T, \mu,\eta}$ be a monad on $\fib$ in \biFibc. 
Then, the triple $\DRF(\mnd) = \ple{\DRF(T),\DRF(\mu) \cdot \ple{\id,\omega^{T,T}},\DRF(\eta) \cdot \ple{\id,\xi^\fib}}$
is a monad on $\GDS\fib$ in \Fib. 
\end{cor}
\begin{proof}
The key observation is that monads in a 2-category \Ct{K} can be seen as lax functors from the trivial 2-category (it has one object, one 1-arrow, and one 2-arrow) into \Ct{K}. 
Then the action of a lax functor \fun{\mathrm{F}}{\Ct{K}}{\Ct{H}} on a monad in \Ct{K} is given by composition of lax functors, which leads to the definition of $\DRF(\mnd)$. 
\end{proof}

Extending also the strength of a strong monad to differential relations is less trivial. 
To this end, we need to further assume that the bifibration $\fib$ has finite products, so that,  from \refItem{prop:gds-lf}{1},  the fibration $\GDS\fib$ has finite products as well, which are necessary to talk about strong monads. 
We also have to assume that the product functor $\dtimes$ on the total category of $\fib$ preserves cocartesian arrows, 
as often happens when dealing with monoidal bifibrations \cite{Schulman08,MelliesZ16}.\footnote{Here the monoidal structure is given just by cartesian product.} 
This is a sensible requirement as the action of $\DRF$ on 1-arrows heavily relies on cocartesian liftings. 
 From a more abstract perspective, this assumption ensures that the bifibrations we are considering are precisely the cartesian objects in $\biFibc$, thus in the following we will refer to them just as \emph{bifibrations with finite products}. 
Then, we can prove the following result. 

\begin{thm}\label{thm:gds-str-mnd}
Let $\fun\fib\E\B$ be a bifibration with finite products 
and $\mnd = \ple{T, \mu,\eta,\str}$ be a strong monad on $\fib$ in \biFibc. 
Then, $\DRF(\mnd)$ is a strong monad on $\GDS\fib$  in \Fib. 
\end{thm}
\begin{proof}
First, we notice that the lax functor $\DRF$ preserves 2-products, that is, 
$\DRF(\fib\otimes\fibq)$ is a 2-product of $\DRF(\fib)$ and $\DRF(\fibq)$, where we use $\otimes$ to denote 2-products to avoid confusion with usual products. 
Recall from \refItem{prop:gds-lf}{1} that $\GDS\fib$ has finite products and that, for objects $\dX,\dY$ in $\GDSC\fib$, we have 
$\dun{\dX\dtimes\dY} = \dun\dX \times \dun\dY$, 
$\dsp{\dX\dtimes\dY} = \dsp\dX \times \dsp\dY$ and 
$\dist{\dX\dtimes\dY} = \lift{\sigma_{\dX,\dY}}{(\dist\dX\dtimes\dist\dY)}$, where 
$\sigma_{\dX,\dY} = \iple{\pi_1,\pi_3,\pi_5,\pi_2,\pi_4,\pi_6}$ is a canonical isomorphism from $\nabla(\dX\dtimes\dY)$ to $\nabla(\dX)\times\nabla(\dY)$, natural in $\dX$ and $\dY$. 
Note also that, $\dist{\dX\dtimes\dY}$ is vertically isomorphic to $\colift{{\sigma'_{\dX,\dY}}}{(\dist\dX\dtimes\dist\dY)}$ where $\sigma'_{\dX,\dY} = \iple{\pi_1,\pi_4,\pi_2,\pi_5,\pi_3,\pi_6}$ is the inverse of $\sigma_{\dX,\dY}$. 

Now, given a 1-arrow \oneAr{T}{\fib}{\fib} we can define 1-arrows
\oneAr{F_T,G_T}{\fib\otimes\fib}{\fib}  given by 
\[
\fb{F_T}(I,J) \eqdef I\times \fb{T} J \quad 
\ft{F_T}(X,Y) \eqdef X\dtimes \ft{T} Y \quad 
\fb{G_T}(I,J) \eqdef \fb{T}(I\times J) \quad 
\ft{G_T}(X,Y) \eqdef \ft{T}(X\times Y)
\]
Then, a strength $\str$ for $T$ in \biFibc is a 2-arrow 
\twoAr{\str}{F_T}{G_T} in \biFibc and so 
\twoAr{\ple{\fb\str,\what\str}}{\ple{\fb{F_T},\what{F_T}}}{\ple{\fb{G_T},\what{G_T}}} is a 2-arrow in \Fib. 
However, this is not a strength for \ple{\fb{T},\what{T}} in \Fib, because $\what\str$ is not a strength for $\what{T}$, as its domain and codomain have not the right form. 
In order to get a proper strength, we need to adjust $\what{\str}$, by composing it with suitable (vertical) natural transformations defined below. 

We can construct a vertical  arrow 
\fun{\theta^T_{\dX,\dY}}{\dX\dtimes\what{T}\dY}{\what{F_T}(\dX,\dY)}
natural in both $\dX$ and $\dY$. 
Indeed, if $\dZ = \what{F_T}(\dX,\dY)$, we have 
\[
\dun\dZ = \dun\dX\times\fb{T}\dun\dY 
\quad 
\dsp\dZ = \dun\dX\times\dsp\dX\times\fb{T}(\dun\dY\times\dsp\dY) 
\quad 
\dist\dZ = \colift{\projar{\fb{F_T}}{\ple{\dX,\dY}}}{(\dist\dX\dtimes\ft{T}\dist\dY)}
\]
Clearly, we have $\dun\dZ = \dun{\dX\dtimes\what{T}\dY}$. 
Then, we have 
$\Der(\theta^T_{\dX,\dY}) = \iple{\pi_1,\pi_2,\pi_1,\pi_3,\pi_4}$ from 
$\dun\dX\times\fb{T}\dun\dY \times\dsp\dX\times \fb{T}(\dun\dY\times\dsp\dY)$ to $\nabla\dZ$ and 
we can construct $\dpr{\theta^T_{\dX,\dY}}$ by the following diagrams and the fact that $\dtimes$ preserves cocartesian arrows in $\fib$: 
\begin{mathpar}
\vcenter{\xymatrix{
  \nabla\dX\times\fb{T}(\nabla\dY) 
    \ar[r]^-{\projar{\fb{F_T}}{\ple{\dX,\dY}}} 
    \ar[d]_-{\id_{\nabla\dX}\times \projar{\fb{T}}{\dY}} 
& \nabla\dZ 
\\ 
  \nabla\dX\times\nabla(\what{T}\dY) 
    \ar[r]^-{\sigma'_{\dX,\what{T}\dY}} 
& \nabla(\dX\dtimes\what{T}\dY) 
    \ar[u]^-{\Der(\theta^T_{\dX,\dY}) \times \id_{\dun\dX\times\fb{T}\dun\dY}}
}} \and 
\vcenter{\xymatrix@C=10ex{
  \dist\dX\dtimes\ft{T}\dist\dY 
    \ar[r]^-{ \coliftar{\projar{\fb{F_T}}{\ple{\dX,\dY}}}{\dist\dX\dtimes\ft{T}\dist\dY} } 
    \ar[d]_-{ \id_{\dist\dX}\dtimes \coliftar{\projar{\fb{T}}{\dY}}{\dist\dY} } 
& \dist\dZ 
\\ 
  \dist\dX\dtimes\dist{\what{T}\dY} 
    \ar[r]^-{ \coliftar{\sigma'_{\dX,\what{T}\dY}}{\dist\dX\dtimes\dist{\what{T}\dY}} } 
& \dist{\dX\dtimes\what{T}\dY} 
    \ar@{..>}[u]_-{ \dpr{\theta^T_{\dX,\dY}} } 
}}
\end{mathpar}

Similarly, there is a vertical arrow 
\fun{\zeta^T_{\dX,\dY}}{\what{G_T}(\dX,\dY)}{\what{T}(\dX\dtimes\dY)}, 
natural in both $\dX$ and $\dY$. 
Indeed, if $\dZ = \what{G_T}(\dX,\dY)$, we have 
\[
\dun\dZ = \fb{T}(\dun\dX\times\dun\dY)  
\quad 
\dsp\dZ = \fb{T}(\dun\dX\times\dsp\dX\times\dun\dY\times\dsp\dY)  
\quad 
\dist\dZ = \colift{\projar{\fb{G_T}}{\ple{\dX,\dY}}}{\ft{T}(\dist\dX\dtimes\dist\dY)} 
\]
Clearly, we have 
$\dun\dZ = \dun{\what{T}(\dX\dtimes\dY)}$. 
Then, we have 
$\Der(\zeta^T_{\dX,\dY}) = \id_{\dun\dZ} \times \fb{T}\iple{\pi_1,\pi_3,\pi_2,\pi_4}$ from $\nabla\dZ$ to $\fb{T}(\dun\dX\times\dun\dY) \times \fb{T}(\dun\dX\times\dun\dY\times\dsp\dX\times\dsp\dY)$, and 
we construct $\dpr{\zeta_{\dX,\dY}}$ by the following diagrams: 
\begin{mathpar}
\vcenter{\xymatrix{
  \fb{T}(\nabla\dX\times\nabla\dY) 
    \ar[r]^-{\projar{\fb{G_T}}{\ple{\dX,\dY}}} 
    \ar[d]_-{\fb{T}\sigma'_{\dX,\dY}} 
& \nabla\dZ 
    \ar[d]_-{\Der(\zeta^T_{\dX,\dY})\times\id_{\dun\dZ}} 
\\
  \fb{T}(\nabla(\dX\dtimes\dY))
    \ar[r]^-{\projar{\fb{T}}{\dX\dtimes\dY}} 
& \nabla(\what{T}(\dX\dtimes\dY)) 
}}  \and 
\vcenter{\xymatrix@C=10ex{
  \ft{T}(\dist\dX\dtimes\dist\dY) 
    \ar[r]^-{ \coliftar{\projar{\fb{G_T}}{\ple{\dX,\dY}}}{\ft{T}(\dist\dX\dtimes\dist\dY)} }
    \ar[d]_-{ \ft{T}\coliftar{\sigma'_{\dX,\dY}}{\ft{T}(\dist\dX\dtimes\dist\dY)} }
& \dist\dZ 
    \ar@{..>}[d]^-{\dpr{\zeta_{\dX,\dY}}}
\\
  \ft{T}(\dist{\dX\dtimes\dY}) 
    \ar[r]^-{ \coliftar{\projar{\fb{T}}{\dX\dtimes\dY}}{\ft{T}(\dist{\dX\dtimes\dY})} }
& \dist{\what{T}(\dX\dtimes\dY)} 
%    \ar@{..>}@/^7pt/[u]^-{\dpr{\zeta^{-1}_{\dX,\dY}}} 
}}
\end{mathpar}

Then, one can easily check that $\theta^T$ and $\zeta^T$ are natural in $T$, that is, 
given 1-arrows \oneAr{S,T}{\fib}{\fib} and a 2-arrow \twoAr{\phi}{S}{T} in \biFibc, we have: 
\[
\theta^T \cdot (\id\dtimes \what{\phi}) = \what{\id\times\phi} \cdot \theta^S 
\quad 
\zeta^T \cdot \what{\phi_{\blank\times\blank}} = \what{\phi}_{\blank\dtimes\blank} \cdot \zeta^S 
\]

Let $\mnd = \ple{T,\mu,\eta,\str}$ be a strong monad on $\fib$ in \biFibc. 
The lifted strength is then defined as follows 
\[
\str'_{\dX,\dY} = \xymatrix@C=6ex{
  \dX\dtimes \what T\dY \ar[r]^-{ \theta^T_{\dX,\dY} }
& \what{F_T}(\dX,\dY) \ar[r]^-{\what\str_{\dX,\dY}}
& \what{G_T}(\dX,\dY) \ar[r]^-{\zeta^T_{\dX,\dY}}
& \what T(\dX\dtimes\dY) 
}
\]
Notice that $\dun{\str'} = \fb{\str}$, because both $\theta^T$ and $\zeta^T$ are vertical, hence 
$\ple{\fb\str,\str'}$ is a well-defined 2-arrow in \Fib. 
The fact this is indeed a strength follows from the strength laws for $\str$ using the following equalities, which can be proved by routine calculations with projections: 
\begin{mathpar} 
\zeta^{\Id} = \id
\and 
\theta^{\Id} \cdot (\id\dtimes \xi^\fib) = \xi^\fib_{\blank\dtimes\blank}
\and 
\omega^{T,T}_{\blank\dtimes\blank} \cdot \what{T}\zeta^T = \zeta^{TT} \circ \omega^{G_T,T} 
\and 
\omega^{F_T,T} \cdot \what{T}\theta^T \cdot \zeta^T_{\blank,\what{T}\blank}
  = \omega^{\Id\otimes T,G_T} \cdot \what{G_T}(\xi^\fib\otimes\id) 
\and 
\omega^{\Id\otimes T,F_T} \cdot \theta^T_{\what\Id\blank,\what{T}\blank} \cdot (\xi^\fib\otimes\id) 
  = \theta^{TT} \cdot (\id\dtimes\omega^{T,T}) 
\end{mathpar} 
We report the case of the monad multiplication. 
\begin{align*}
\what\mu_{\dX\dtimes\dY} &\circ \omega^{T,T}_{\dX\dtimes\dY} \circ \what{T}\str'_{\dX,\dY} \circ \str'_{\dX,\what{T}\dY} 
% expand def of str' 
  = \what\mu_{\dX\dtimes\dY} \circ \omega^{T,T}_{\dX\dtimes\dY} \circ \what{T}\zeta^T_{\dX,\dY} \circ \what{T} \what\str_{\dX,\dY} \what{T}\theta^T_{\dX,\dY} \circ \str'_{\dX,\what{T}\dY} 
\\
% first omega equation 
  &= \what\mu_{\dX\dtimes\dY} \circ \zeta^{TT}_{\dX,\dY} \circ \omega^{G_T,T}_{\dX,\dY} \circ \what{T} \what\str_{\dX,\dY} \circ \what{T}\theta^T_{\dX,\dY} \circ \str'_{\dX,\what{T}\dY} 
\\
% naturalit of zeta and omega 
  &= \zeta^T_{\dX,\dY} \circ \what{\mu_{\blank\times\blank}}_{\dX,\dY} \circ \what{T\str}_{\dX,\dY} \circ \omega^{F_T,T}_{\dX,\dY} \circ \what{T}\theta^T_{\dX,\dY} \circ \str'_{\dX,\what{T}\dY} 
\\
% expand def of str'
  &= \zeta^T_{\dX,\dY} \circ \what{\mu_{\blank\times\blank}}_{\dX,\dY} \circ \what{T\str}_{\dX,\dY} \circ \omega^{F_T,T}_{\dX,\dY} 
      \circ \what{T}\theta^T_{\dX,\dY} \circ \zeta^T_{\dX,\what{T}\dY} \circ \what{\str}_{\dX,\dY} \circ \theta^T_{\dX,\what{T}\dY}
\\
% second omega equation 
  &= \zeta^T_{\dX,\dY} \circ \what{\mu_{\blank\times\blank}}_{\dX,\dY} \circ \what{T\str}_{\dX,\dY} \circ \omega^{\Id\otimes T, G_T}_{\dX,\dY} 
      \circ \what{G_T}(\xi^\fib_{\dX}, \id_{\what{T}\dY}) \circ \what{\str}_{\dX,\what{T}\dY} \circ \theta^T_{\dX,\what{T}\dY} 
\\
% exchange law 
  &= \zeta^T_{\dX,\dY} \circ \what{\mu_{\blank\times\blank}}_{\dX,\dY} \circ \what{T\str}_{\dX,\dY} \circ \omega^{\Id\otimes T, G_T}_{\dX,\dY} 
     \circ \what{\str}_{\dX,\what{T}\dY} \circ \theta^T_{\what{\Id}\dX,\what{T}\dY} \circ (\xi^\fib_{\dX} \dtimes \id_{\what{T}\what{T}\dY}) 
\\
% naturality of omega 
  &= \zeta^T_{\dX,\dY} \circ \what{\mu_{\blank\times\blank}}_{\dX,\dY} \circ \what{T\str}_{\dX,\dY} \circ \what{\str_{\blank,T\blank}}_{\dX,\dY} 
     \circ \omega^{\Id\otimes T,F_T}  \circ \theta^T_{\what{\Id}\dX,\what{T}\dY} \circ (\xi^\fib_{\dX} \dtimes \id_{\what{T}\what{T}\dY}) 
\\
% strength law for mu 
  &= \zeta^T_{\dX,\dY} \circ \what{\str}_{\dX,\dY} \circ \what{\id\times\mu}_{\dX,\dY} 
     \circ \omega^{\Id\otimes T,F_T}  \circ \theta^T_{\what{\Id}\dX,\what{T}\dY} \circ (\xi^\fib_{\dX} \dtimes \id_{\what{T}\what{T}\dY}) 
\\
% third omega equation 
  &= \zeta^T_{\dX,\dY} \circ \what{\str}_{\dX,\dY} \circ \what{\id\times\mu}_{\dX,\dY} 
     \circ \theta^{TT}_{\dX,\dY} \circ (\id_\dX \dtimes \omega^{T,T}_{\dY}) 
\\
% naturality of theta 
  &= \zeta^T_{\dX,\dY} \circ \what{\str}_{\dX,\dY} \circ \theta^T_{\dX,\dY} 
     \circ (\id_\dX \dtimes \what{\mu}_\dY)  \circ (\id_\dX \dtimes \omega^{T,T}_{\dY}) 
\\
% def of str' 
  &= \str'_{\dX,\dY} \circ (\id_\dX \dtimes (\what{\mu}_\dX \circ \omega^{T,T}_\dX))
     \circ (\id_\dX \dtimes \what{\mu}_\dY)  \circ (\id_\dX \dtimes \omega^{T,T}_{\dY}) \tag*{\qedhere}
\end{align*}

\end{proof}

In other words, these results provide us with a tool for building a differential extension of a (strong) monad $\fb\mnd$ on the base category of a bifibration $\fib$ 
starting from a usual extension $\ft\mnd$ along $\fib$ which preserves cocartesian arrows. 
Therefore, to build a differential extension of $\fb\mnd$, we can use existing techniques to lift it, obtaining a monad on the bifibration $\fib$, and then apply our construction. 

We conclude this section by instantiating this technique to a special class of bifibrations, namely, bifibrations of weak subobjects 
as defined in \cref{ex:wsub}. 
The resulting construction applies to (strong) monads on categories with finite products and weak pullbacks, and provides a differential version of the Barr extension, that we dub a \emph{differential Barr extension}.

First, we show that the construction of the bifibration of weak subobjects extends to a (strict) 2-functor.\footnote{That is, a 2-functor strictly preserving horizontal composition. }
Let us denote by \Catcwp the 2-category of categories with finite products and weak pullbacks, functors, and natural transformations.\footnote{Note that we do not require functors to preserve either weak pullbacks or finite products. } 
Given a functor \fun{F}{\CC}{\D} in \Catcwp, define 
\fun{\ohat F}{\WS\CC}{\WS\D} as 
$\ohat F \ple{X,[\alpha]} \eqdef \ple{FX,[F\alpha]}$ and 
$\ohat F f \eqdef F f$. 
It is easy to check that this is indeed a functor. 
Moreover, note that if \fun{f}{\ple{X,[\alpha]}}{\ple{Y,[\beta]}} is cocartesian in $\wSub[\CC]$, that is, $[\beta] = [f\alpha]$, 
then \fun{\ohat F f}{\ple{FX,[F\alpha]}}{\ple{FY,[F\beta]}} is cocartesian in $\wSub[\D]$, because $[F\beta] = [Ff\circ F\alpha]$. 
Consider now a natural transformation \nt{\phi}{F}{G} in \Catcwp and define 
\nt{\ohat \phi_{\ple{X,[\alpha]}}}{\ohat F}{\ohat G} as 
$\ohat \phi_{\ple{X,[\alpha]}} \eqdef \phi_X$. 
This is well-defined because, if \fun{\alpha}{A}{X}, then we have $\phi_X\circ F\alpha = G\alpha \circ \phi_A$, by naturality of $\phi$. 
Thus, we define 
\[
\wSub[F] \eqdef \ple{F,\ohat F} 
\qquad 
\wSub[\phi] \eqdef \ple{\phi,\ohat \phi}
\] 
and we get the following result

\begin{prop}\label{prop:wsub-fun}
\fun{\Psi}{\Catcwp}{\biFibc} is a strict 2-functor. 
\end{prop}

 This fact immediately implies that $\Psi$ maps monads in \Catcwp to monads in \biFibc. 
Furthermore, $\Psi$ trivially preserves products, thus it maps cartesian objects in \Catcwp to cartesian objects in \biFibc, that is, bifibrations with finite products. 
Therefore, we get that $\Psi$ preserves strengths as well. 
Then, using \cref{thm:gds-str-mnd} we get the following result.  

\begin{cor}\label{thm:gds-mnd2}
Let $\mnd$ be a (strong) monad on a category \CC with finite products and weak pullbacks. 
Then, $\ft{\DRF(\wSub[\mnd])}$ is a differential extension of $\mnd$ along $\wSub[\CC]$. 
\end{cor}
\begin{proof}
If $\mnd$ is a (strong) monad on \CC, since $\Psi$ is a strict product preserving 2-functor (\cref{prop:wsub-fun}), 
then $\wSub[\mnd]$ is a (strong) monad on $\wSub[\CC]$ and so, by \cref{thm:gds-str-mnd}, $\DRF(\wSub[\mnd])$ is a (strong) monad on $\GDS{\wSub[\CC]}$, as needed. 
\end{proof}

\begin{exa}\label{ex:coupling}
Let \ple{T,\mu,\eta} be a monad on \Set,  which is necessarily strong. 
The \emph{coupling-based lifting} of \ple{T,\mu,\eta} to $\GDSC{\Sub[\Set]}$ \cite{DalLagoG22} maps an object \ple{X,V,R} to 
\ple{TX,T(X\times V),\tau(R)} where 
$\ple{x,v,y}\in\tau(R)$ iff there exists $\varphi \in TR$ such that $T\pi_1(T\iota_R(\varphi)) = x$, $T\iple{\pi_1,\pi_2}(T\iota_R(\varphi)) = v$, and $T\pi_3(T\iota_R(\varphi)) = y$. 
Here, \fun{\iota_R}{R}{X\times V \times X} denotes the inclusion function, and the element $\varphi$ is called a (ternary) coupling, borrowing the terminology from optimal transport \cite{Villani/optimal-transport/2008}. 
Basically, this lifting states that monadic elements $x$ and $y$ are related with monadic distance $v$ if and only if we can find a coupling $\varphi$ that projected along the first and third component gives $x$ and $y$, respectively, and projected along the first two components gives the monadic distance $v$. 
Relying on the equivalence between $\Sub[\Set]$ and $\wSub[\Set]$, 
we get that the above lifting of \ple{T,\mu,\eta} is an instance of the differential Barr extension
\fun{\what{\ohat T}}{\GDSC{\WS\Set}}{\GDSC{\WS\Set}}. 
Indeed, for an object $\ple{X,V,[\alpha]}$ in $\GDSC{\wSub[\Set]}$  we can choose a canonical representative \fun{\iota_R}{R}{X\times V \times X} such that $[\alpha] = [\iota_R]$, where $R$ is the image of $\alpha$ and $\iota_R$ is the inclusion function. 
\end{exa}

%% file: conclu.tex
\section{Conclusion and Future Work}

% Fibrational accounts of logical relations have been given since the seminal work by Hermida 
% \cite{Hermida93} giving rise to a rich literature dealing with important programming language 
% features, such as parametricty and polymorphism \cite{johann1, johann2, reddy} and computational effects 
% \cite{DBLP:journals/lmcs/KatsumataSU18,DBLP:journals/iandc/Katsumata13,DBLP:conf/popl/Katsumata14,DBLP:conf/csl/Katsumata05, 10.1007/978-3-319-71237-6_17,DBLP:journals/entcs/KammarM18,GoubaultLasotaNowak/MSCS/2008}. 
% All the aforementioned accounts, however, define logical predicates and relations on 
% program denotations, rather than on program themselves. 

We have shown how fibrations can be used to give a uniform account of operational logical relations 
for higher-order languages with generic effects and a rather liberal operational semantics. 
Our framework encompasses both traditional, set-based logical relations and logical relations on 
non-cartesian-closed categories --- such as one of measurable spaces --- as well as 
the recently introduced differential logical relations. In particular, 
our analysis sheds a new light on the mathematical foundation of both pure and effectful 
differential logical relations.  
Further examples of logical relations that can be described in our framework 
include classic Kripke logical relations \cite{mitchell} (take fibrations of 
poset-indexed relations).
%as well as logical relations for information flow 
% (the latter could be approached both as suitable Kripke logical 
% relations or by considering a shallow semantics on the category of 
% classified sets~\cite{DBLP:journals/pacmpl/Kavvos19}).
Additionally, since differential logical relations can be used to reason about 
nontrivial notions such as program sensitivity and cost analysis \cite{DalLagoG22}, our 
framework can be used for reasoning about the same notions too.  

Even if general, the vehicle calculus of this work lacks some important programming language features, 
such as recursive types and polymorphism. Our framework being operational, the authors 
suspect that the addition of polymorphism should not be problematic, whereas 
the addition of full recursion may require to come up with abstract notions of step-indexed 
logical relations \cite{AppelMcAllester/TOPLAS/2001,DBLP:journals/corr/abs-1103-0510}. 
In general, the proposed framework looks easily extensible and adaptable. 
To reason about a given calculus, one should pick a fibration
with a logical structure supporting the kind of analysis one is interested in and 
whose base category has enough structure to model its interactive behaviour.  
For instance, 
in this work we just need products and a monad on the base category but, e.g., to model sum types, we would need also 
coproducts to describe case analysis. 
On the logical side, we have just considered standard intuitionistic connectives but, e.g., 
to support more quantitative analysis, one may need to use fibrations supporting linear connectives and modalities \cite{Schulman08,MelliesZ15,LicataS17,DagninoR21,DagninoP22}. 

Besides the extension of our framework to richer languages and features, 
an interesting direction for future work is to formally relate our results with 
general theories of denotational logical relations. Particularly relevant for that seems 
the work by Katsumata \cite{DBLP:phd/ethos/Katsumata05} who observes that the closed structure 
of base categories is not necessarily essential.
Another interesting direction for future work is the development of fibrational theories of
coinductive reasoning for higher-order languages. Fibrational accounts of coinductive techniques 
have been given in the general setting of coalgebras \cite{DBLP:conf/csl/BonchiPPR14,DBLP:conf/concur/Bonchi0P18}, 
whereas general accounts of coinductive techniques for higher-order languages have been 
obtained in terms of relational reasoning \cite{Gavazzo/ICTCS/2017,DalLagoGavazzoLevy/LICS/2017,Gavazzo/LICS/2018,dal-lago/gavazzo-mfps-2019,DBLP:conf/esop/LagoG19,DBLP:journals/tcs/LagoGT20,DBLP:conf/fscd/LagoG21,DBLP:phd/basesearch/Gavazzo19,DBLP:journals/pacmpl/LagoG22a}.  
It would be interesting to see whether these two lines of research could be joined in our fibrational 
framework. That may also be a promising path to the development of conductive differential reasoning.

% \cite{Hermida93,johann1, johann2,DBLP:journals/lmcs/KatsumataSU18,DBLP:journals/iandc/Katsumata13,DBLP:conf/popl/Katsumata14,DBLP:conf/csl/Katsumata05}
% \emph{reflexive graphs} \cite{OHearnT95,RobinsonR94,reddy,HermidaRR14}, and 
% \emph{factorisation systems} \cite{DBLP:journals/entcs/KammarM18,DBLP:journals/entcs/HughesJ02,GoubaultLasotaNowak/MSCS/2008}. 
% The byproduct of all of that is a general, 
% highly modular
% theory of denotational logical relations that has been successfully applied to a large array of language 
% features, ranging from parametricty and polymorphism \cite{johann1, johann2, reddy} to computational effects 
% \cite{DBLP:journals/lmcs/KatsumataSU18,DBLP:journals/iandc/Katsumata13,DBLP:conf/popl/Katsumata14,DBLP:conf/csl/Katsumata05, 10.1007/978-3-319-71237-6_17,DBLP:journals/entcs/KammarM18,GoubaultLasotaNowak/MSCS/2008}.